\documentclass[reqno,11pt]{amsart}
\usepackage{amsmath, latexsym, amsfonts, amssymb, amsthm, amscd}
\usepackage{graphics,epsf,psfrag}
\usepackage{graphics,epsf,psfrag}
\usepackage{stmaryrd}
\newcommand{\ext}{\mathrm{ext}}
\renewcommand{\bar}{\overline}

\newcommand{\lint}{\llbracket}
\newcommand{\rint}{\rrbracket}
\newcommand{\odd}{\mathrm{odd}}
\newcommand{\even}{\mathrm{even}}
\numberwithin{equation}{section}

\newtheorem{theorema}{Theorem}

\newtheorem{theorem}{Theorem}[section]
\newtheorem{lemma}[theorem]{Lemma}
\newtheorem{proposition}[theorem]{Proposition}

\newtheorem{rem}[theorem]{Remark}

\newcommand{\bad}{\mathrm{bad}}
\newcommand{\good}{\mathrm{good}}
\newcommand{\Diam}{\mathrm{Diam}}
\newcommand{\dd}{\mathrm{d}}
\newcommand{\ind}{\mathbf{1}}
\newcommand{\restrict}{\!\!\upharpoonright}
\newcommand{\Supp}{\mathrm{Supp}}
\newcommand{\larg}{\mathrm{large}}

\renewcommand{\tilde}{\widetilde}
\renewcommand{\hat}{\widehat}
\newcommand{\cc}{\complement}

\newcommand{\cX}{{\ensuremath{\mathcal X}} }
\newcommand{\cA}{{\ensuremath{\mathcal A}} }
\newcommand{\cB}{{\ensuremath{\mathcal B}} }

\newcommand{\cP}{{\ensuremath{\mathcal P}} }

\newcommand{\cH}{{\ensuremath{\mathcal H}} }
\newcommand{\cC}{{\ensuremath{\mathcal C}} }
\newcommand{\cN}{{\ensuremath{\mathcal N}} }
\newcommand{\cL}{{\ensuremath{\mathcal L}} }

\newcommand{\cZ}{{\ensuremath{\mathcal Z}} }

\newcommand{\bP}{{\ensuremath{\mathbf P}} }
\newcommand{\bQ}{{\ensuremath{\mathbf Q}} }
\newcommand{\bE}{{\ensuremath{\mathbf E}} }

\DeclareMathSymbol{\leqslant}{\mathalpha}{AMSa}{"36} % nicer `smaller or equal'
\DeclareMathSymbol{\geqslant}{\mathalpha}{AMSa}{"3E} % nicer `larger or equal'
\DeclareMathSymbol{\eset}{\mathalpha}{AMSb}{"3F}     % nicer `emptyset'
\newcommand{\sumtwo}[2]{\sum_{\substack{#1 \\ #2}}} % sum with 2 lines
\newcommand{\limtwo}[2]{\lim_{\substack{#1 \\ #2}}}     % \lim with 2 lines
\newcommand{\prodtwo}[2]{\prod_{\substack{#1 \\ #2}}}     % product 2 lines
\newcommand{\bbH}{{\ensuremath{\mathbb H}} }

\newcommand{\bbN}{{\ensuremath{\mathbb N}} }

\newcommand{\bbR}{{\ensuremath{\mathbb R}} }

\newcommand{\bbT}{{\ensuremath{\mathbb T}} }

\newcommand{\bbZ}{{\ensuremath{\mathbb Z}} }

\newcommand{\gb}{\beta}
\newcommand{\gep}{\varepsilon}       % \ge already exists...

\newcommand{\gG}{\Gamma}

\newcommand{\gO}{\Omega}
\newcommand{\gl}{\lambda}
\newcommand{\gL}{\Lambda}

\makeatletter
\def\captionfont@{\footnotesize}
\def\captionheadfont@{\scshape}

\long\def\@makecaption#1#2{%
  \vspace{2mm}
  \setbox\@tempboxa\vbox{\color@setgroup
    \advance\hsize-6pc\noindent
    \captionfont@\captionheadfont@#1\@xp\@ifnotempty\@xp
        {\@cdr#2\@nil}{.\captionfont@\upshape\enspace#2}%
    \unskip\kern-6pc\par
    \global\setbox\@ne\lastbox\color@endgroup}%
  \ifhbox\@ne % the normal case
    \setbox\@ne\hbox{\unhbox\@ne\unskip\unskip\unpenalty\unkern}%
  \fi
  \ifdim\wd\@tempboxa=\z@ % this means caption will fit on one line
    \setbox\@ne\hbox to\columnwidth{\hss\kern-6pc\box\@ne\hss}%
  \else % tempboxa contained more than one line
    \setbox\@ne\vbox{\unvbox\@tempboxa\parskip\z@skip
        \noindent\unhbox\@ne\advance\hsize-6pc\par}%
\fi
  \ifnum\@tempcnta<64 % if the float IS a figure...
    \addvspace\abovecaptionskip
    \moveright 3pc\box\@ne
  \else % if the float IS NOT a figure...
    \moveright 3pc\box\@ne
    \nobreak
    \vskip\belowcaptionskip
  \fi
\relax
}
\makeatother
\def\writefig#1 #2 #3 {\rlap{\kern #1 truecm
\raise #2 truecm \hbox{#3}}}

\newcommand{\tf}{\textsc{f}}

\title[Wetting and layering for Solid-on-Solid I]{Wetting and layering for Solid-on-Solid I: \\
 Identification of the wetting point and critical behavior}

\author{Hubert Lacoin}
\address{
  IMPA, Institudo de Matem\'atica Pura e Aplicada, Estrada Dona Castorina 110
Rio de Janeiro, CEP-22460-320, Brasil. 
}

\begin{document}

\begin{abstract}
We provide a complete description of the low temperature wetting transition for the two dimensional Solid-On-Solid model. 
More precisely we study
the integer-valued field $(\phi(x))_{x\in \bbZ^2}$,
associated associated to the energy functional
$$V(\phi)=\gb \sum_{x\sim y}|\phi(x)-\phi(y)|-\sum_{x}\left( h\ind_{\{\phi(x)=0\}}-\infty\ind_{\{\phi(x)<0\}} \right).$$
It is known since the pioneering work Chalker \cite{cf:chal} 
of that for every $\gb$, there exists $h_{w}(\gb)>0$ delimiting a transition between a delocalized phase ($h<h_{w}(\gb)$)
where the proportion of points at level zero vanishes, and a localized phase ($h>h_{w}(\gb)$) where this proportion is positive.
We prove in the present paper that for $\gb$ sufficiently large we have
$$h_w(\gb)= \log \left(\frac{e^{4\gb}}{e^{4\gb}-1}\right).$$
Furthermore we provide a sharp asymptotic for the free energy at the vicinity of the critical point:
 we show that close to $h_w(\gb)$, the free energy is approximately piecewise affine and that the points of discontinuity for the derivative 
 of the affine approximation forms a geometric sequence accumulating on the right of $h_w(\gb)$.
 This asymptotic  behavior provides a strong evidence for the conjectured existence of countably many ``layering transitions'' 
 at the vicinity of the critical point,
 corresponding to jumps for the typical height of the field.
\\[10pt]
\\[10pt]
\end{abstract}

\maketitle

\tableofcontents

\newpage
\section{Introduction}

\subsection{Motivation: the wetting problem for Ising interfaces}

Consider the Ising model at low temperature on a cube large cube $\lint 1, N\rint^3$ with minus  boundary condition 
 on the bottom face and plus on the other five faces. If the inverse temperature $\gb$ is sufficiently large, the only macroscopic interface that appears in the system is the one separating pluses and minuses appears at the vicinity of the bottom face.

\medskip

This interface wants to minimize its surface tension, and thus has an incentive to remain close to the bottom of the cube. By
doing so it also loses some degrees of freedom as it feels the constraint of the boundary. Due to this energy/entropy competition, 
the typical distance of the interface to the bottom face diverges when $N$ goes to infinity, a phenomenon known as entropic repulsion (rigorously proved in \cite{cf:FP}).

\medskip

We can take to problem one step further and add some positive external magnetic fields $h$ acting only on the layer of the cube which is adjacent to the bottom face. 
Heuristically, it is quite clear that a very large magnetic fields forces the interface to stick to the bottom face.
It is a more delicate issue, however,
to decide whether this occurs for a magnetic field with an arbitrary small intensity, 
or if we have a non-trivial transition in $h$ between a phase where entropic repulsion persists and a localized one.
This question, as well as the qualitative description of such a phase transition, has given rise to a rich literature, and we refer to the recent survey
by Ioffe and Velenik \cite{cf:IV} for a comprehensive bibliography on the subject.

\medskip

While the analogous problem in dimension $2$ can be solved explicitly (there is a closed expression for the critical value of the magnetic field \cite{cf:Ab} 
and the behavior of the interface in the localized and the delocalized phase seems to be well understood \cite[Section 2.1.2]{cf:IV}) 
the understanding of the problem in dimension $3$ (and higher) is far from complete  \cite[Section 2.1.3]{cf:IV}. 

\medskip

One reason for being so is that when $d=3$, the behavior of interfaces between bulk phases (far from the boundary) is rigorously understood only
at very low temperature.
In that regime, it is known that
Ising interfaces are rigid, in the sense that they are microscopically  flat apart from local perturbations
 \cite{cf:Dob}.
The existence of the wetting transition in this regime was established in \cite{cf:FP} together with an explicit bound for 
the value of the magnetic field where the wetting 
transition should occur. We let $h_w(\gb)$ denote the critical value of $h$ at which this transition occurs (a rigorous characterization of this value is given in \eqref{defhw}).

\medskip

However, the predicted complete low temperature picture goes beyond the existence of a wetting transition: When the magnetic fields varies,
the system should undergo a (countably) infinite sequence of phase transitions corresponding to changes of the height for the interface. The corresponding critical points 
should accumulate towards the left of  $h_{w}(\gb)$ (see e.g. \cite{cf:ADMM, cf:Bas} for evidences of this phenomenon). 
Making this rigorous is a very challenging task and for this reason some attention has been brought to a simplified version of the problem
called Solid-On-Solid (SOS).

\subsection{The wetting problem for the Solid-on-Solid model}

The SOS model is obtained  by neglecting  interaction with microscopic clusters and assuming that the interface is a graph of a function
(see \eqref{defSOS} below for a formal definition).
While these simplification should not alter the qualitative behavior of the interface, it makes the model much more tractable:
Rigorous results have been proved concerning the existence of a roughening transition \cite{cf:FS} (this transition is also 
conjectured to occur for the three dimensional Ising Model), 
and a very detailed picture of the entropic repulsion phenomenon 
at low temperature has been given in \cite{cf:PMT}.

\medskip

The problem of wetting
for SOS has been considered in the literature with a particular focus on the case of dimension $1+2$ 
(the case corresponding to the three dimensional Ising model). 
The existence of a transition was proved in proved by Chalker \cite{cf:chal} with bounds on $h_w(\gb)$ valid at all temperature which 
implies in particular that entropic repulsion persists with the presence of a small magnetic fields.
\medskip

In \cite{cf:ADM}, a perturbational approach based on Cluster Expansion techniques provides some evidence for the infinite sequence layering transition.
More precisely, it is proved that for every $n$, at sufficiently low temperature (which depends on $n$), 
there exists an interval of values of $h$ for which the field localizes at height
$n$.
Related results concerning infinite volume Gibbs states as well as regularity of the free energy in the corresponding regions are also obtained.
However, as the requirement on the temperature in the main results of \cite{cf:ADM} depends on the value of $n$, it does not imply the occurrence of 
an infinite sequence of phase transition, nor does 
it allow to identify the critical wetting line $h_w(\gb)$.

\medskip

In the present paper, we provide a deeper understanding of  the wetting transition at low temperature, with the following results:
\begin{itemize}
 \item [(A)] We identify the critical wetting line, proving that for large values of $\gb$
 $$h_w(\gb)= \log \left(\frac{e^{4\gb}}{e^{4\gb}-1}\right).$$
 It corresponds to the lower bound derived by Chalker \cite[Equation (4a)]{cf:chal}. 
\item [(B)] We prove that close to the critical point, the free energy is asymptotically equivalent to a 
function which is affine on intervals delimited by a geometric progression which accumulates on the right of $h_w(\gb)$.
\end{itemize}
Point (B) is very much related to the conjectured existence of sequence of layering transition that should accumulate on the right of $h_{w}(\gb)$ (discussed e.g.\ in 
\cite{cf:ADMM, cf:ADM,cf:IV}). Let us briefly explain this connection:
Our asymptotic for the free energy is obtained by trying to find an optimal strategy to benefit from the pinning without paying too much entropic cost.
The strategies we consider are of the following type: The interface stabilize at an height $n$ above the bottom face, and visit the interaction zone
mostly with some rare spikes pointing 
downward. We obtain a sharp asymptotic by taking the supremum over over all such strategies (that is over the localization height $n$) 
as evidenced by Equation \eqref{uds} .
Each affine part in the asymptotic equivalent displayed in \eqref{uds} corresponds to a given value of $n$, and
the meeting points between two affine parts correspond thus to a transition between two localization strategies.

\medskip

We believe, in agreement with existing conjectures in the literature \cite{cf:ADMM, cf:ADM, cf:Bas}
that these angular point do not appear only on the asymptotic equivalent but also on the free energy curve. Thus they should  correspond to  discontinuity point for the asymptotic contact fraction, which is the signature of first-order phase transition. 
Between these layering transition, the free energy should be analytic in both parameters $h$ and $\gb$.
This aspect of the problem is to be developed in a second paper \cite{cf:HL}.

\section{Model and results}

\subsection{The Solid on Solid Model on $\bbZ^d$}

While we are mostly interested in the two dimensional case, the proof presented in this paper can be adapted for arbitrary dimension (see the discussion in Section \ref{possibext}). Hence,
we introduce the problem in the more general framework.
Consider $\gL$ a finite subset of  $\bbZ^d$ (equipped with its usual lattice structure) and let $\partial \gL$ denote its external boundary 
$$\partial \gL:= \{ x \in \bbZ^d \setminus \gL \ : \ \exists y \in \gL, \ x \sim y \}.$$
Given $\phi\in \gO_{\gL}:= \bbZ^\gL$ and $n\in \bbZ$ we define the Hamiltonian for SOS with boundary condition $n$ as,
\begin{equation}\label{defhamil}
\cH^{n}_\gL(\phi):= \frac{1}{2}\sumtwo{x, y \in \gL}{ x \sim y} |\phi(x)-\phi(y)|
+\sumtwo{x\in \gL, y\in \partial\gL}{x\sim y} |\phi(x)-n|.
\end{equation}
Given $\gb>0$, we define
the SOS measure with boundary condition $n$, $\bP^{n}_{\gL,\gb}$ on $\gO_{\gL}$ by
\begin{equation}\label{defSOS}
\bP^{n}_{\gL,\gb}(\phi):= \frac{1}{\mathcal Z_{\gL,\gb}}e^{-\gb\cH^{n}_\gL(\phi)} \quad \text{where} \quad
\mathcal Z_{\gL,\gb}:=\sum_{\phi\in \gO_{\gL}}  e^{-\gb\cH^n_\gL(\phi)}.
\end{equation}
Note that, by translation invariance, $\mathcal Z_{\gL,\gb}$ does not depend on $n$.
For readability, we drop the superscript $n$ in the notation in the special case $n=0$.

\medskip

\noindent Observe that if $\gL^{(1)}$ and $\gL^{(2)}$ are disjoint  we 
 have 
 \begin{equation}\label{hicup}
 \cH_{\gL^{(1)}\cup \gL^{(2)}}(\phi)\le  \cH_{\gL^{(1)}}(\phi)+\cH_{\gL^{(2)}}(\phi)
 \end{equation}
which yields immediately 
\begin{equation}\label{hicup2}
 \mathcal Z_{\gL^{(1)}\cup \gL^{(2)},\gb}\ge \cZ_{\gL^{(1)},\gb}\cZ_{\gL^{(2)},\gb}.
\end{equation}
This property implies (see e.g.\ \cite[Section 3.2]{cf:Vel} for a proof in the case of the Ising model) 
the existence of the following limit
\begin{equation}
\limtwo{|\gL|\to \infty}{|\partial\gL|/|\gL|\to 0} \frac{1}{|\gL|}\log  \cZ_{\gL,\gb}=\tf(\gb),
\end{equation}
taken over any sequence of finite sets $(\gL_N)_{N\ge 1}$ such that the ratio between the cardinality of 
$\gL_N$ and that of its boundary $\partial \gL_N$ vanishes. We refer to $\tf(\gb)$ as the free energy.

\medskip

To clarify notation, in the remainder of the paper, we consider the limit along the sequence 
$\gL_N:=\lint 1, N\rint^2$ (using the notation $\lint a,b \rint=[a,b]\cap \bbZ$ ). We write $\cZ_{N,\gb}$ for $\cZ_{\gL_N,\gb}$ and adopt a similar notation for the  other quantities.

\medskip

When $\gb$ is sufficiently large, it has been shown in \cite{cf:BM} that there exists a unique infinite volume Gibbs state corresponding 
to zero boundary condition. To state the result in full, we need to introduce some classic terminology.

\medskip

We say that a function $f: \ \gO_{\bbZ^d}:=(\bbZ)^{\bbZ^d}\to \bbR$ is local of there exists $(x_1,\dots,x_k)$  and $\tilde f:  \bbZ^{k}\to \bbR$ 
such that $f(\phi)=\tilde f(\phi(x_1),\dots,\phi(x_k)).$ The minimal  choice for the set of indices $\{x_1,\dots, x_k\}$ is called the support of $f$ ($\Supp(f)$).
With some abuse of notation, whenever $\gL$ contains the support of $f$,  we extend $f$ to $\gO_{\gL}$ in the obvious way.
An event is called local if its indicator function is a local function.

\medskip

\noindent For $A$ and $B$ two finite subsets of $\bbZ^2$, we define 
\begin{equation}\label{disset}
d(A,B):=\min_{x\in A, y\in B} |x-y|,
\end{equation}
where $|\ \cdot \ |$ denote the $\ell_1$ distance. In \cite{cf:BM}, is has been shown that  $\bP_{\gL,\gb}$ converges exponentially fast to some infinite volume Gibbs 
measure in the low temperature regime. The following result is a consequence of the proof of the main theorem in \cite{cf:BM} (it can also be deduced as a consequence of 
the first proposition in \cite{cf:KP}).
\begin{theorema}[From \cite{cf:BM}]\label{infinitevol}
For any $d\ge 2$, there exists $\gb_0(d)>2$ such that 
for any $\gb>\gb_0$, 
there exists a measure $\bP_{\gb}$  $\gO_{\bbZ^d}$ such that that for every local function 
$f$, and every $\gL$ which contains the support of $f$. 
\begin{equation}\label{decayz}
| \bE_{\gL,\gb}[f(\phi)]-\bE_{\\gb}[f(\phi)] |\le 4  |\Supp(f)|\|f\|_{\infty} e^{- d(\partial \gL, \Supp(f))}.
\end{equation}
\end{theorema}

\subsection{The wetting problem for the SOS model}

For $\phi \in \gO_{\gL}$ and $A\subset \bbZ$ (or $\bbR$), we set 
$$\phi^{-1}(A):=\{ x\in \gL \ : \ \phi(x)\in A\}.$$
We sometimes omit the brackets, and write $\phi^{-1}A$, to make the notation lighter.
Given $h\in \bbR$ we consider $\bP^{h}_{\gL,\gb}$ which is a modification of $\bP_{\gL,\gb}$ 
where the interface $\phi$ is constrained to remain positive 
and gets an energetic reward $h$ for each contact with $0$.
We use the notation $\gO^+_{\gL}:=(\bbZ_+)^{\gL}$ where $\bbZ_+:=\bbZ\cap[0,\infty)$ and define
\begin{equation}\label{def}
\bP^{h}_{\gL,\gb}(\phi):= \frac{1}{\mathcal Z^h_{\gL,\gb}}e^{-\gb\cH_\gL(\phi)+ h|\phi^{-1}\{0\}|}
 \quad \text{where } \mathcal Z^h_{\gL,\gb}:=\sum_{\phi\in \gO^+_{\gL}}  e^{-\gb\cH_\gL(\phi)+h|\phi^{-1}\{0\}|}.
\end{equation}
We also consider $\bP^{n,h}_{\gL,\gb}$ the variant with boundary condition $n\in \bbN$, where $\cH_{\gL}$ is replaced by $\cH^n_{\gL}$.
We want to study  the localization transition in $h$ for $\bP^{h}_{\gL,\gb}$ which appears in the limit when $\gL\to \infty$.
For the same reasons as for the SOS partition function (recall \eqref{hicup2}), 
$\log  \cZ^h_{\gL,\gb}$ is superadditive for disjoint union and thus the free energy
\begin{equation}
\tf(\gb,h):=\lim_{N\to \infty} \frac{1}{N^2}\log \mathcal Z^{h}_{N,\gb},
\end{equation}
is well defined.
Furthermore 
the function $h\mapsto \tf(\gb,h)$ is non-decreasing and convex in $h$ (as a limit of non-decreasing convex function).
At points where $\tf(\gb,h)$ is differentiable, the  convexity allows to exchange the positions of limit and derivative, thus the derivative in $h$ corresponds to the 
asymptotic contact fraction
\begin{equation}
 \partial_h \tf(\gb,h)=\lim_{N\to \infty} \frac{1}{N^2} \bE^h_{\gL,\gb}[|\phi^{-1}(0)|].
\end{equation}
It was established in \cite{cf:chal} (in different but equivalent terms) that the critical value $h_w(\gb)$ separating the localized phase  ($\partial_h \tf(\gb,h)>0$) from the delocalized one  ($\partial_h \tf(\gb,h)=0$) is given by 
\begin{equation}\begin{split}\label{defhw}
 h_w(\gb)&:= \sup\{ h\in \bbR : \tf(\gb,h)= \tf(\gb)\}\\
 &=\inf\{h \in \bbR \ :  \tf(\gb,h)>0\}.
\end{split}\end{equation}
Additionally, it was proved that for every $\gb>0$
\begin{equation}\label{chalbound}
 \log \left(\frac{e^{4\gb}}{e^{4\gb}-1}\right)\le h_w(\gb)\le \log \left(\frac{16(e^\gb+1)}{e^\gb-1} \right),
\end{equation}
showing in particular that $h_w(\gb)>0$.
The aim of this paper is to determine the value of $h_w(\gb)$,
and to identify some properties of the system at the vicinity of the critical point:
In particular we are interested in identifying the asymptotic behavior of $\tf(\gb,h_w(\gb)+u)$ for small positive values of $u$.

\subsection{Main result}
We focus now on the case $d=2$.
To state our main result, we need to introduce two quantities related to asymptotic probability for observing certain types of ``peak'' in $\phi$.
We let ${\bf 0}$ and ${\bf 1}$ denote the vertices $(0,0)$ and $(1,0)$ respectively and define, for $\gb>\gb_0$
\begin{equation}\label{singledouble}
\begin{split}
\alpha_1(\gb)&:= \lim_{n\to \infty} e^{4\gb}  \bP_{\gb} \left[ \phi({\bf 0})\ge  n\right],\\
\alpha_2(\gb)&:= \lim_{n\to \infty} e^{6\gb}  \bP_{\gb} \left[  \min( \phi({\bf 0}), \phi({\bf 1})) \ge n\right].
\end{split}
\end{equation}
The existence and finiteness of those limits is part of Proposition \ref{sendingout} (the existence of $\alpha_1$ 
is given in \cite[Lemma 2.4]{cf:PLMST}).
\noindent For better readability of the formulas we introduce the quantities
 $$\bar\tf(\gb,u)=\tf\left(\gb,\log \left(\frac{e^{4\gb}}{e^{4\gb}-1}\right)+u\right)-\tf(\gb)$$
and
$$J=J(\gb):= e^{-2\gb}.$$

\begin{theorem}\label{main}
If $\gb>\gb_0$ (of Theorem \ref{infinitevol}), we have 
$$h_w(\gb)=\log\left(\frac{e^{4\gb}}{e^{4\gb}-1}\right).$$
Furthermore
\begin{equation}\label{lequivdelamor}
\bar \tf(\gb,u)\stackrel{u\to 0+}{\sim} F(\gb,u)
\end{equation}
where  
\begin{equation}\label{uds}
F(\gb,u):=\max_{n\in \bbZ_{+}}\left( \alpha_1 J^{2n}u-  \frac{2 \alpha_2(J^3-J^4)}{1-J^3} J^{3n} \right).
\end{equation}
\end{theorem}

\begin{rem} Note that the value of $n$ that maximizes the l.h.s.\ in \eqref{uds} is given by
\begin{equation}
n(u):=\left \lceil (2\gb)^{-1} \log \left( \frac{2\alpha_2(1-J^3)}{\alpha_1 (1+J)(1-J^4)}  \right)- \frac{\log u}{2\gb}\right \rceil-3.
\end{equation}
In particular $F(\gb,u)$ is affine by parts and we have 
\begin{equation}\label{letroi}
F(\gb,u)=  u^3 g_{\gb}\left( \log_J u\right),
\end{equation}
where $g_{\gb}$ is a $1$-periodic function which is bounded away from $0$ and $\infty$. 
The critical exponent $3$ appearing in \eqref{letroi} is a not universal and is inherited from the lattice structure, 
see Section \ref{trans} for more on this issue.
 \end{rem}

\begin{rem}
Let us be more specific on the assumption required for $\gb$.
To prove that $h_w(\gb)=\log\frac{e^{4\gb}}{e^{4\gb}-1}$, we only need
the following  to hold 
\begin{equation}\label{finite}
 \sum_{\{\gamma \ :   \ {\bf 0}\in\bar \gamma \}}e^{-\gb|\tilde \gamma|}<\infty,
\end{equation}
where the sum is performed over contours which contain ${\bf 0}$ (see Section \ref{contour} where the relevant definitions are introduced).
Moreover under assumption \eqref{finite} alone, we can prove that there exists a constant $C$ such that 
$$\forall u\in(0,1], \quad C^{-1}u^3 \le \bar \tf(\gb,u)\le C u^3.$$

\medskip

To obtain the sharp asymptotics, we use Theorem \ref{infinitevol} and thus we need $\gb\ge \gb_0$. On top of that, we also use the fact that $\gb>\gb_1$ where 
\begin{equation}\label{finiteprim}
 \gb_1= \inf\left\{ \gb \ : \ \sum_{\{\gamma :  {\bf 0}\in\bar \gamma \}}e^{-\gb|\tilde \gamma|}<\infty \right\},
\end{equation}
(which is almost equivalent but possibly more strict than \eqref{finite}).
It is a classical exercise to show that $\gb_1\in (\log 2, \log 3)$. While no explicit bound given in \cite{cf:BM} for $\gb_0$,
the issue of is discussed in \cite[p 493]{cf:KP}, and it appears that $\gb_0\ge 2$ and thus $\gb_0>\gb_1$.
\end{rem}
\begin{rem}
Derivative of $F(\gb,u)$ appears to be discontinuous at 
$$u_n:=\frac{2 \alpha_2}{\alpha_1}\frac{J^{2n+2}(1+J)(1-J^3)}{1-J^4}.$$
While we believe that discontinuity point appears also on the free energy curve, to mark the layering transitions, it is not expected that they coincide exactly 
with $u_n$. We conjecture however that $\partial_u \bar \tf(\gb,u)$ present a sequence of discontinuity points $(u^*_n)_{n\ge 1}$, which is asymptotically equivalent to $u_n$ when 
$n$ tends to infinity (see also Section \ref{slayer}).
\end{rem}

\subsection{Possible extensions of the result}\label{possibext}

We have chosen for simplicity to present our result in $\bbZ^2$, but it can  extend to any finite dimensional transitive lattice.
For two dimensional lattices, the proof works verbatim, provided adequate conventions  
are adopted for the contour decomposition.

\medskip

The case of higher dimension (for simplicity let us say $\bbZ^d$) is a bit more delicate as the contour representation introduced in Section \ref{contour} 
is specific to dimension $2$. However this is not essential to the proof and  
we can replace contours e.g.\ with clusters of edges of non-zero gradient as done in \cite{cf:BM}.

\subsubsection{Transitive planar lattices}\label{trans}

Let $\bbT_2$ and $\bbH_2$ denote the usual triangular lattice and hexagonal lattice respectively: 
The vertices of $\bbT_2$ are the complex numbers of the form $a+be^{i\pi/3}$ with $a,b\in \bbZ$ and edges link vertices at (Euclidean) distance $1$.
The lattice $\bbH_2$ is obtained in the same manner but only keeping the vertices such that $a+2b\ne 2 \ [3]$.
The SOS model can be defined on those lattices  in the same manner as for $\bbZ^2$.

\medskip

To adapt the proof of Theorem \ref{main} to these cases, the only ingredient we need is a contour representation similar to the one introduced in
Section \ref{contour}.
For $\bbT_2$ there is no possible ambiguity when defining contour in the dual lattice, but the situation is a bit more delicate for
$\bbH_2$, since the dual lattice has degree $6$ (and thus $3$ contour line can meet). We let the reader check that things work out using an adequate convention 
to define the contour lines, similar the one described below Figure \ref{linked}.

\medskip

If we let $\bP^{\#}_{\gb}$  denote the infinite volume limit (with $\#=\bbT$ or $\bbH$) of the SOS measure with $0$ boundary condition 
(which exist for $\gb$ sufficiently large),
we define 
\begin{equation}\begin{split}
\alpha^{\bbT}_1:= \lim_{n\to \infty}e^{6\gb}  \bP^{\bbT}_{\gb}[\phi({\bf 0})],\quad 
\alpha^{\bbT}_2:= \lim_{n\to \infty}e^{10\gb}  \bP^{\bbT}_{\gb}[\min(\phi({\bf 0}), \phi({\bf 1}))\ge n], \\
\alpha^{\bbH}_1:= \lim_{n\to \infty}e^{3\gb}  \bP^{\bbT}_{\gb}[\phi({\bf 0})],\quad 
\alpha^{\bbH}_2:= \lim_{n\to \infty}e^{4\gb}  \bP^{\bbT}_{\gb}[\min(\phi({\bf 0}), \phi({\bf 1}))\ge n]. \\
\end{split}
\end{equation}
Using the same notation as in the $\bbZ^2$ case with corresponding superscript and setting 
$$\bar \tf^{\#}(\gb,u)=\tf^{\#}(\gb,h_{w}^{\#}(\gb)+u)-\tf^{\#}(\gb)$$
we obtain the following result.

\begin{theorem}
When $\gb> \gb_0(\#)$ sufficiently large, we have, 
$$h^{\bbT}_w(\gb)=\log\left(\frac{e^{6\gb}}{e^{6\gb}-1}\right) \quad \text{ and } \quad h^{\bbH}_w(\gb)=\log\left(\frac{e^{3\gb}}{e^{3\gb}-1}\right)$$
and furthermore  we have
\begin{equation}\label{lequivdelamor2}
\bar \tf^{\#}(\gb,u)\stackrel{u\to 0+}{\sim} F^{\#}(\gb,u)
\end{equation}
where  
\begin{equation}\label{uds2}\begin{split}
F^{\bbT}(\gb,u)&:=\max_{n\in \bbZ_{+}}\left( \alpha^{\bbT}_1 J^{3n}u-  \frac{3\alpha^{\bbT}_2(J^5-J^6)}{1-J^5} J^{5n} \right),\\
F^{\bbH}(\gb,u)&:=\max_{n\in \bbZ_{+}}\left( \alpha^{\bbH}_1 J^{3n/2}u-  \frac{3\alpha^{\bbH}_2(J^{2}-J^{3})}{2(1-J^2)} J^{2n} \right).
\end{split}
\end{equation}
\end{theorem}
The proof being identical to that of Theorem \ref{main}, we leave to the reader the computation need to obtain \eqref{uds2}. 
Our main motivation for displaying this generalization  is to underline that 
the exponent $3$ appearing in \eqref{letroi} is not universal and inherited from the lattice structure (more precisely it is determined by the 
ratio of the length of the two shortest contours). For instance we have 
\begin{equation}
\bar \tf^{\bbT}(\gb,u)\stackrel{u\to 0+}{\asymp} u^{5/2} \quad  \text{ and } \quad \bar \tf^{\bbH}(\gb,u)\stackrel{u\to 0+}{\asymp} u^{4}. 
\end{equation}

\subsubsection{Higher dimension}

The generalization of the result to higher dimension is a bit more delicate since one has to abandon the 
idea of using the contour representation,
but it can be still be performed.
In that case, for $\gb$ sufficiently large, we have 
\begin{equation}\label{critd}
h^{d}_w(\gb):=\log \left(\frac{e^{2d\gb}}{e^{2d\gb}-1}\right)
\end{equation}
and the asymptotic equivalent for $$\bar \tf^{d}(\gb,u):=\tf^{d}(\gb,h_{w}^{d}(\gb)+u)-\tf^d(\gb)$$
is given by 
\begin{equation}
F^{d}(\gb,u):=\max_{n\in \bbZ_{+}}\left( \alpha^{d}_1 J^{dn}u-  \frac{d\alpha^{d}_2(J^{2d-1}-J^{2d})}{1-J^{2d-1}} J^{(2d-1)n} \right).
\end{equation}
where  we defined
\begin{equation}\begin{split}\label{limzits}
\alpha^{d}_1&:= \lim_{n\to \infty}e^{2d\gb}  \bP^{d}_{\gb}[\phi({\bf 0})], \\
\alpha^{d}_2&:= \lim_{n\to \infty}e^{(4d-2)\gb}  \bP^{d}_{\gb}[\min(\phi({\bf 0}), \phi({\bf 1}))\ge n],
\end{split}
\end{equation}
with ${\bf 0}=(0,\dots,0)$ and ${\bf 1}=(1,0,\dots,0)$. The letter $\bP^{d}_{\gb}$ denotes the infinite volume probability distribution.
In particular we have $$\bar \tf^{d}(\gb,h^d_w(\gb)+u)\stackrel{u\to 0+}{\asymp} u^{\frac{2d-1}{d-1}}.$$

\medskip

It is known however that when $d\ge 3$ the SOS surface is rigid for every value of $\gb$ \cite{cf:BFL},
and hence it seems reasonable to hope that \eqref{critd} holds for all $\gb$.  Giving any prediction about the critical behavior for small $\gb$ is a more challenging task.

\begin{rem}
The case of SOS in dimension $1$ is quite different and corresponds to the wetting of random walk investigated in \cite{cf:Fisher} (see also the first chapters of 
\cite{cf:GB} for an introduction).
In that case we have $h_w(\gb)=\log (1+e^{-\gb})$ for every $\gb>0$, and there exists a closed expression for the free-energy $\tf(\gb,h)$ which is analytic in $\gb$ and $h$ outside of the critical line.
\end{rem}

\begin{rem}
The question of wetting can be and has been posed for other kinds of surface models, in particular for those when $\phi$ takes value in $\bbR^{\gL}$
and the potential is of the form 
$$-h\ind_{\{\phi(x)\in[0,1]\}}+\infty \ind_{\{\phi(x)<0\}}.$$
These models are not expected to display layering transitions.
In \cite{cf:CV} it has been shown that 
there is a non trivial transition ($h_w(\gb)>0$) when $\phi$ is distributed like the lattice massless free field in $\bbZ^2$. 
There is no particular reason to hope that the value of $h_w(\gb)$ can be determined in that case, and obtaining some qualitative result about the phase transition 
seems a very challenging task.

\medskip

For the lattice free field in higher dimension ($d\ge 3$) it has been shown that there is no wetting transition ($h_w(\gb)=0$)  \cite{cf:BZ}. 
Moreover in that case the critical behavior of the free energy
and the vicinity has been identified \cite{cf:GL}.

\end{rem}

\subsection{About layering transitions}\label{slayer}

While the asymptotics we find for the free energy and the corresponding proofs give some indications about the occurrence of layering transitions, 
they do not provide a complete picture of the phenomenon.
Let us describe shortly here what we believe should occur, when $\gb$ is sufficiently large

\medskip

\begin{itemize}
\item[(A)]
There exists a decreasing sequence $(h^*_n(\gb)_{n\ge 1})$, with $\lim_{n\to \infty} h^*_n=h_w(\gb)$
 such that $h\mapsto \tf(\gb,h)$ is analytic on $\bbR\setminus \{h^*_n\}_{n\ge 1}$, and 
and  $\partial_h\bar \tf(\gb,h)$ is discontinuous at $h^*_n$ for every $n$.
\item[(B)] When $h\in [h^*_{n+1},h^*_{n})$,  the measure $\bP^{h}_{\gL,\gb}$ converges towards an infinite volume limit $\bP^{n,h}_{\gb}$ which
 \textit{percolates at level} $n$. We mean that under $\bP^{h}_{\gb}$, the level set 
$\{ x \ : \ \phi(x)=n \}$ contains a unique infinite connected component of positive density and the other level sets only display finite connected components. This convergence also holds 
for $\bP^{m,h}_{\gL,\gb}$ for any $m\in \bbN$.
 \item[(C)] When $h\in (h^*_{n+1},h^*_{n})$, $\bP^{n,h}_{\gb}$ is the unique translation invariant infinite volume Gibbs state, but when
 When $h=h*_n$ there exists (at least) two translation invariant Gibbs states $\bP^{n-1,h^*_n}_{\gb}$ and $\bP^{n,h^*_n}_{\gb}$ which respectively percolates at level $n-1$ and $n$.
 The   sequence $\bP^{m,h^*_n}_{\gL,\gb}$ converges towards $\bP^{n-1,h^*_n}_{\gb}$ when $|\gL|\to \bbZ^2$ if $m\le n-1$ and 
 towards $\bP^{n,h^*_n}_{\gb}$ if $m\ge n$.
 \end{itemize}
 
A step towards this result was performed in \cite{cf:ADM}: analyticity and  the existence of the infinite volume measure at level $n$ were proved to hold in the interval 
$(J^{n+2+\gep},J^{n+3-\gep})$ for $\gb\ge \gb_n$, where $\gb_n$ diverges when $n$ goes to infinity.  
A scheduled sequel to the present paper aims to bring more of this conjecture on rigorous ground \cite{cf:HL}.

\medskip

Results with a similar flavor where obtained in \cite{cf:DM, cf:CM} for the so-called prewetting problem, where a bulk positive magnetic field is present.

\section{A heuristic picture of the phase diagram}

In the present section, we introduce a reformulation of the problem which aims to 
give a better intuition both on the result and the  proof of Theorem \ref{main}.
In Section \ref{rewriite}, we reinterpret the partition function as that of a different interface model which does not include a positivity constraint and is easier to analyze.
We make use of this alternative representation in Section \ref{houra} to provide a heuristic for our main result.

\subsection{Rewriting the partition function to drop the positivity constraint} \label{rewriite}

We can rewrite the partition function for the wetting as a partition function on the larger set of unconstrained trajectories $\gO_{\gL}$, 
by adding some penalty term for visiting the negative half space.
Recalling \eqref{def}, for any $\gG\subset \gL$ we set
\begin{equation}\label{ZZplus}
\cZ^+_{\gG}:=\cZ^0_{\gG,\gb}=\sum_{\phi\in \gO^+_{\gG}} \exp\left(-\gb \cH_\gL(\phi)\right).
\end{equation}
and
$$H(\gG):=\log \mathcal Z^+_{\gG,\gb},$$
(the dependence in $\gb$ is omitted from the notation in both case for better readability).
We can rewrite the partition function in the following manner.
\begin{lemma}\label{represent}
We have 
\begin{equation}
\cZ^{h}_{\gL,\gb}=\sum_{\phi\in \gO_{\gL}} e^{-\gb \cH_{\gL}(\phi)+ h|\phi^{-1}(\bbZ_-)|-H(\phi^{-1}(\bbZ_-))}.
\end{equation}
\end{lemma}

\begin{proof}
 We let $\phi_+(x)= \max(\phi(x),0)$ and $\phi_-(x)=\max(-\phi(x),0)$ denote respectively the positive and negative part of $\phi$.
Observe that the map $\phi \mapsto (\phi_+,\phi_{-})$ is bijective from $\gO_{\gL}$ 
to the set of pairs of function with disjoint supports 
 $$\left\{ (\phi_+,\phi_{-})\in \gO^+_{\gL} \ : \ \forall x \in \gL, \ \phi_+(x)\phi_{-}(x)=0\right\}.$$
Considering $\phi_-$ as a function defined 
on $\phi^{-1}_+\{0\}$, it is straight forward to check from the definition that
$$\cH_{\gL}(\phi)=\cH_{\gL}(\phi_+)+\cH_{\phi^{-1}_+\{0\}}(\phi_-).$$
 As a consequence we have 
 \begin{multline}
\sum_{\phi\in \gO_{\gL}} e^{-\gb \cH_{\gL}(\phi)+ h|\phi^{-1}(\bbZ_-)|-H(\phi^{-1}(\bbZ_-))}
\\=\sum_{\phi_+\in \gO^+_{\gL}}  e^{-\gb \cH_{\gL}(\phi_+)+ h|\phi^{-1}_+\{0\}|-H(\phi^{-1}_+\{0\})} 
\sum_{\phi_-\in \gO^+_{\phi^{-1}_+\{0\}}}e^{-\gb \cH_{\phi^{-1}_+\{0\}}(\phi_-)}
\\=\sum_{\phi_+\in \gO^+_{\gL}}  e^{-\gb \cH_{\gL}(\phi_+)+ h|\phi^{-1}_+\{0\}|-H(\phi^{-1}_+\{0\})}  \cZ^+_{\phi^{-1}_+\{0\}}\\
=\sum_{\phi_+\in \gO^+_{\gL}}  e^{-\gb \cH_{\gL}(\phi_+)+ h|\phi^{-1}_+\{0\}|}= \cZ^{h}_{\gL,\gb}.
\end{multline}
\end{proof}

\noindent The proof displayed above implies also that the distribution of $\bP^{h}_{\gL,\gb}$ can be obtained by first by sampling 
a field in $\gO_{\gL}$ with respect to the alternative measure $\tilde \bP^{h}_{\gL,\gb}$ defined by
\begin{equation}\label{latilde}
\tilde \bP^{h}_{\gL,\gb}(\phi):= \frac{1}{\cZ^{h}_{\gL,\gb}} e^{-\gb \cH_{\gL}(\phi)+ h|\phi^{-1}(\bbZ_-)|-H(\phi^{-1}(\bbZ_-))},
\end{equation}
and then taking the positive part of $\phi$.
While the measure $\tilde \bP^{h}_{\gL,\gb}$ might seem less natural than    $\bP^{h}_{\gL,\gb}$, it turns out to be easier to analyze because
it has a positive density with respect to $\bP_{\gL,\gb}$ whose behavior is better understood (cf.\ Theorem \ref{infinitevol}, Lemma \ref{geom} and Lemma \ref{restrict}).

\medskip

Yet, to make use of this alternative representation for the wetting model, we need to state few basic properties for $H$.
We say that two finite subsets of $\bbZ^2$, $\gL^{(1)}$ and $\gL^{(2)}$ are \textit{separated} if they are disjoint and are not connected by any edges
(this is equivalent to $\gL^{(1)}\cap \partial \gL^{(2)}=\emptyset$). 

\medskip

A first obvious observation is that if $\gL^{(1)}$ and $\gL^{(2)}$ are separated, we have 
$$H(\gL^{(1)}\cup \gL^{(2)})=H(\gL^{(1)})+H(\gL^{(2)}).$$
Thus $H(\gG)$ can be computed by summing over all maximal connected components of $\gG$.
The following result gives us the contributions of components of size one and two and an estimate for that of size three and larger.

\begin{lemma}\label{lehagga}
The function $H_{\gb}=H$ defined on the finite subsets of $\bbZ^d$ satisfies the following 
\begin{itemize}
\item[(i)] For $x\in \bbZ^2$, $H\{x\}=\log\left(\frac{1}{1-J^2}\right)$.
\item[(ii)] If  $x,y\in \bbZ^2$, are such that  $x\sim y$ then
$$H\{x,y\}=2\log\left(\frac{1}{1-J^2}\right)+ \log \left(\frac{1-J^4}{1-J^3}\right).$$
\item[(iii)] For a connected set $\gG$ with $|\gG|\ge 2$ there exists positive constants $c_1(\gb)$ and $c_2(\gb)$ which are such that  
\begin{equation}\label{evlat}
c_1|\gG|\le  H(\gG)- |\gG| \log\left(\frac{1}{1-J^2}\right)\le  c_2|\gG|.
\end{equation}

\end{itemize}
\end{lemma}
\begin{proof}
To prove $(i)$ we simply observe that 
$$\cZ^+_{\{x\}}= \sum_{z=0}^{\infty} e^{-4\gb z}=\frac{1}{1-J^2}.$$
We obtain $(ii)$ by decomposing on the possible value for
 $u=\min(\phi_x,\phi_y)$, $v=|\phi_y-\phi_x|$. 
Considering that two location are possible for the maximum (when $\phi_x\ne \phi_y$) we obtain
$$\cZ^+_{\{x,y\}}= \sum_{u\ge 0}e^{-6\gb u}\left( 1+2\sum_{v\ge 1} e^{-4\gb v}\right)= \frac{(1+J^2)}{(1-J^3)(1-J^2)}.$$

\noindent We are going to prove $(iii)$ with constants (which we do not try to optimize)
$$c_1(\gb)= \frac{1}{6}  \log \left(\frac{1-J^4}{1-J^3}\right) \text{ and } c_2(\gb):= 2\log \left(1+J\right).$$
For the  lower bound we notice (this is a consequence of \eqref{hicup}) that given two disjoint sets $\gL_1$ and $\gL_2$, we have
\begin{equation}\label{ineq}
\cZ^+_{\gL}\ge \cZ^+_{\gL_1}\cZ^+_{\gL_2}
\end{equation}
so that $H$ is superadditive.
Then by immediate induction, we see that any connected set of cardinality $n$ can be split into $\lceil (n-1)/4 \rceil$ connected sets of cardinality $2$ and 
$n-2\lceil (n-1)/4 \rceil$ singletons.
As a consequence of superadditivity and of $(i)-(ii)$, for any $\gG$ we have
\begin{multline}
H(\gG)\ge  \left \lceil \frac{|\gG|-1}{4} \right\rceil \log \left(\frac{1-J^4}{1-J^3}\right) -
\left(|\gG|-2\left\lceil \frac{|\gG|-1)}4 \right\rceil \right)\log \left( \frac{1}{1-J^2} \right)\\
\ge |\gG|\left(  \log\left(\frac{1}{1-J^2}\right)+ \frac{1}{6}  \log \left(\frac{1-J^4}{1-J^3}\right) \right).
\end{multline}
For the upper bound, we notice that for each $x\in \gG$, extending $\phi$ to the boundary by setting $\phi(y)=0$ for $\phi\in \partial \gG$ we have 

\begin{multline}
\cZ^+_{\gG,\gb}\le \cZ_{\gG,\gb}=  \sum_{\phi\in \gO_{\gG}} \left(\prod_{x\in \gG} e^{-\frac{\gb}{2}\sum_{y\sim x}|\phi_x-\phi_y|}\right)
 \exp\bigg(-\frac{\gb}{2}\sumtwo{y\in \gG,z\in \partial \gG}{y\sim z}|\phi_y|\bigg)\\
\le \sum_{\phi\in \gO_{\gG}} \prod_{x\in \gG} e^{-\frac{\gb}{2}\sum_{y\sim x}|\phi_x-\phi_y|}.
\end{multline}
Setting $\bar \phi_x:= \frac{1}{2d}\phi_y$
we have for each $x\in \gG$, 
$$\frac{\gb}{2}\sum_{y\sim x}|\phi_x-\phi_y|\ge 2\gb |\phi_x-\bar \phi_x|,$$ 
and hence
\begin{equation}
 \sum_{\phi_x\in \bbZ} e^{-\frac{\gb}{2}\sum_{y\sim x}|\phi_x-\phi_y|}\le  \sum_{\phi_x\in \bbZ} e^{-2\gb|\phi_x-\bar \phi_x|}
 \le\frac{1+J}{1-J}.
\end{equation}
Taking the product over all $x\in \gG$ this yields 
\begin{equation}
\cZ^+_{\gG,\gb}\le \left(\frac{1+J}{1-J}\right)^{|\gG|}
\end{equation}
and hence the desired upper bound.

\end{proof}

\subsection{Identifying the localization strategy} \label{houra}

A consequence of Lemma \ref{lehagga} is that the function $\bar H$ defined by 
$$\bar H(\gG):=H(\gG)- |\gG|\log \left(\frac{1}{1-J^2}\right),$$
is non-negative.
For that reason, it is convenient to consider the change of variable
$$u:= h- \log\left( \frac{1}{1-J^2}\right).$$
With this notation as a consequence of Lemma \ref{represent} (recall that we use the subscript $N$ when considering $\gL=\gL_N:=\lint 1,N\rint^2$) we have
\begin{equation}\label{youpee}
 \frac{\cZ^{h}_{N,\gb}}{\cZ_{N,\gb}}= \bE_{N,\gb}\left[e^{u |\phi^{-1}(\bbZ_{-})|- \bar H (\phi^{-1}(\bbZ_{-}))} \right].
\end{equation}
Thus, under $\bP^{h}_{N,\gb}$ (recall \eqref{latilde}), if $u\le 0$ trajectories only gets penalized when visiting the lower half-space and trivially
$(\cZ^{h}_{N,\gb}/\cZ_{N,\gb})\le 1$. Thus we recover the lower bound of Chalker \cite{cf:chal} displayed in \eqref{chalbound},
$$h_w(\gb)\ge -\log (1-J^2).$$

\medskip

If $u$ is positive  but small (say $u\le c_1(\gb)/2$ with $c_1(\gb)$ taken from Equation \eqref{evlat})
then the effect one visits of the lower half-space depend on the size of the corresponding cluster of $\phi^{-1}(\bbZ_{-})$ it belongs to.
From Lemma \ref{lehagga}, 
a visit to the lower half-space is rewarded by an amount $u$ when being performed by an isolated site ($\bar H\{x\}=0$), 
and penalized, (with a penalty whose value is comprised between $c_1(\gb)-u \ge c_1(\gb)/2$ and $c_2(\gb)$), 
if at least one of the neighbors also visits the lower half-space.

\medskip

To get an overall positive energetic contribution, we need to cook up a strategy, 
which at a moderate entropic cost makes the vast majority of the visits come from isolated sites.

\medskip

 In view of \eqref{singledouble}, under $\bP_{\gL,\gb}$, high level sets are mostly composed of single peaks. Hence  we decide to change 
 the boundary condition for the field from $0$ to some large $n\ge 0$ in order to make 
 $\phi^{-1}(\bbZ_{-})$ look like one of these level sets. The cost for changing the boundary condition is proportional 
 only to the size of the boundary and thus does not affect the free energy
 (see Lemma \ref{ojko}).
By  translation  and reflection invariance, this change of boundary, is equivalent to replace 
 $\phi^{-1}(\bbZ_{-})$ by  $\phi^{-1}[n,\infty)$.
 
\medskip

Let us now try to find a good  approximation for the structure of the set  $\phi^{-1}[n,\infty)$.
Given $n\ge 1$ (the definition also  extends to $n\le 0$), $\phi\in \gL_N$, and $x\in \gL_N$,
we let $q(\phi,x,n)$ denote the size of the maximal connected component in $\phi^{-1}[n,\infty)$ containing $x$,
that is 
\begin{equation}\label{defq}
q(\phi,x,n):=\max \left\{ \  |\gG| \ : \ x\in \gG,\  \gG \text{ is connected and } \forall y\in \gG, \phi(y)\ge n\right\},
\end{equation}
with the convention that $q(\phi,x,n)=0$ if $\phi(x)\le n-1$.

\medskip

Let us observe that \eqref{singledouble} implies that, for  $x\sim y$, when the boundary is sent to infinity,
the respective probability for $\{q(\phi,x,n)=1\}$ and $\{q(\phi,x,n)=q(\phi,y,n)\ge 2\}$ can be approximated by $\alpha_1 J^{2n}$
 and $\alpha_2 J^{3n}$.
We prove later in the paper (see Proposition \ref{rourou}) that the probability of appearance for cluster of size $3$ or larger is of a smaller order.
 Thus it seems plausible that for small values of $u$,
we can obtain a good approximation of the partition function by only keeping clusters of size one and two 
in the Hamiltonian, giving reward $u$ for clusters of size one, and penalty $\bar H\{x,y\}$ for clusters of size two. 

\medskip

As correlations are rapidly decaying for SOS at low temperature (cf.\ \eqref{decayz}).
 we  decide to push the approximation one step further by assuming that the cluster are ``independently distributed'',
by making the following approximation
\begin{multline}\label{wouhou}
  \bE_{N,\gb}\left[e^{u |\phi^{-1}[n,\infty)|- \bar H (\phi^{-1}[n,\infty))} \right]\\
  \approx 
  \bE\left[  \exp\left( u \sum_{x\in \gL_N} V_x - \sumtwo{x,y\in \gL_N}{x\sim y}  \log \left(\frac{1-J^4}{1-J^3}\right) W_{x,y} \right)\right].
   \end{multline}
  where $(V_x)_{x\in \gL}$ and  $(W_{x,y})_{x,y\in \gL}$ are independent Bernoulli variables
 and for all $x,y$ 
 \begin{equation*}
  \bP[V_x=1]=\alpha_1 J^{2n}, \text{ and }   \bP[W_{x,y}=1]=\alpha_2 J^{3n}.
\end{equation*}
The fact that \eqref{wouhou} might be a valid approximation for some value of $n$ is nothing straight-forward. 
We inferred that the distribution of the clusters of $\phi^{-1}[n,\infty)$ looks like IID,
but it could be that the main contribution to the partition function comes from a very atypical event conditioned to which the distribution of 
clusters has very different characteristics:
this is in fact this is what happens if $n$ is not chosen in the optimal way.
Most of the challenge lies in showing that provided $n$ is well chosen, no such thing happens.

\medskip

\noindent The r.h.s.\ of \eqref{wouhou} is equal to 
\begin{equation}
 \left(1+\alpha_1 J^{2n}(e^u-1)\right)^{N^2}  \left(1- \frac{\alpha_2(J^4-J^3)}{1-J^3}J^{3n}\right)^{2(N^2-N-1)} .
\end{equation}
Taking the $\log$, applying Taylor's formula and dividing by $N^2$ this yields 
\begin{equation}
 \bar \tf(\gb,u)\approx \alpha_1 J^{2n}u-\frac{2\alpha_2(J^4-J^3)}{1-J^3}J^{3n}+O(u^{2}J^{2n}+J^{6n}).
\end{equation}
As the optimal value for $n$ grows when $u$ gets small, the third term is clearly negligible and thus we can infer that \eqref{lequivdelamor} should hold.

\subsection{Organization of the paper}

We could not find a direct path to transform this heuristic into a proof, but both the proof for the lower bound and the upper bound are based on the intuition exposed above.
For the sake of exposition, and also because they require less strict assumptions, we decided to prove first  rougher bounds, 
which allow to identify the 
right order of magnitude for the free energy, before going for the sharper results.

\medskip

\noindent The remainder of the paper is organized as follows:

\medskip

\noindent In Section \ref{prelim}, we present various classic tools and results used in the study of SOS such as monotonicity/FKG inequality (Section \ref{monot}) 
 the contour decomposition (Section \ref{contour}) and Peierls argument (Section \ref{Peierls}).
The proof of the more technical results are deferred to the appendix.

\medskip

\noindent Sections \ref{eloweb} and \ref{nozigo}, are dedicated to prove lower bounds on the free energy.
In Section \ref{eloweb} we prove a bound of order $u^3$ under the minimal assumption \eqref{finite}, and we improve it to a sharp upper bound in Section \ref{nozigo}.

\medskip

\noindent In Section \ref{upeb} we prove an upper bound of order $u^3$ valid in full generality, and improve it in a sharp upper bound under more restrictive 
assumptions in 
Section \ref{sub}.
The different sections are not independent and are meant be read in order.

\section{Preliminaries}\label{prelim}

We present in this section various general tools needed for the  study of the distribution $\bP_{N,\gb}$, and its infinite volume limit.

\subsection{Order and positive correlation}\label{monot}

The set $\gO_{\gL}$ is naturally equipped with an order defined as follows 
$$\phi\le \phi'   \quad \Leftrightarrow \quad \forall x\in \gL,\  \phi(x)\le \phi'(x).$$
Using this order we can define a notion of increasing function ($f$ is increasing if $\phi\le \phi'\Rightarrow f(\phi)\le f(\phi')$) and of increasing event 
($A$ is increasing if the function $\ind_A$ is).
We say that a measure $\mu$ on $\gO_{\gL}$ dominates another measure $\mu'$ (we write $\mu \succcurlyeq \mu'$) if for any increasing function $f$
$$\mu(f(\phi))\ge \mu'(f(\phi))$$

\medskip

It can be easily verified that the distribution given by \eqref{defSOS} satisfies Holley's condition \cite[Equation (7)]{cf:Holley}.
And thus  $\bP^{n}_{\gb,\gL}$ satisfies the FKG inequality which states that for any increasing event $A$ (the opposite inequality holds for decreasing events)
\begin{equation}\label{FKG}
 \bP^{n}_{\gb,\gL}[ \ \cdot \  | \ A] \succcurlyeq   \bP^{n}_{\gb,\gL}.
\end{equation}
We will also use the FKG inequality in some other context for a product measure (in that case the inequality is referred to as FKG-Harris inequality) 
when studying the distribution of contour sets
(see the proof of Lemma \ref{restrict} below).

\subsection{Contour representation}\label{contour}

The contour decomposition of $\phi$ describes how the function $\phi$ can be reconstructed from the knowledge of  its level lines.
The formalism of this section is section is strongly inspired by the one found in \cite{cf:PLMST, cf:PMT} but also borrows ingredients from 
\cite{cf:ADM}.

\medskip

We let $(\bbZ^2)^*$ denote the dual lattice of $\bbZ^2$ (dual edges cross that of $\bbZ^2$ orthogonally in their midpoints).
Two adjacent edges  $(\bbZ^2)^*$ meeting at $x^*$ of are said to be \textit{linked} if they both lie on the same side of the line making an angle $\pi/4$ with the horizontal 
and passing through $x$ (see Figure \ref{linked}).

We define a \textit{contour sequence} to be  a finite sequence $(e_1,\dots,e_n)$ of distinct edges of $(\bbZ^2)^*$ which satisfies:
\begin{itemize}
 \item [(i)] For any $i=\lint 1,n-1\rint$, $e_i$ and $e_{i+1}$ have a common end point $(\bbZ^2)^*$, $e_1$ and $e_{|\tilde \gamma|}$ also have a common end point.
 \item [(ii)] If for $i\ne j$, if $e_i$, $e_{i+1}$, $e_j$ and $e_{j+1}$ meet at a common end point then  $e_i$, $e_{i+1}$ are linked and so are $e_j$ and $e_{j+1}$ (with the convention that $n+1=1$).
\end{itemize}
A \textit{geometric contour} $\tilde \gamma:=\{e_1,\dots,e_{|\tilde \gamma|}\}$ is a set of edges that forms a contour sequence when displayed in the right order.
The cardinality $|\tilde \gamma|$ is called the length of the contour. 
A \textit{signed contour} which we call simply \textit{contour}, is a couple composed of a geometric contour and a $\pm$ sign. 
With some abuse of notation, when a contour $\gamma$ is mentioned, $\tilde \gamma$ will be used to denote the geometric contour associated to $\gamma$. We let 
$\gep(\gamma)$ denote the sign of $\gamma$.

\medskip

We let $\bar \gamma$, the interior of $\gamma$ denote the set of vertices of $\bbZ^2$ enclosed by $\tilde\gamma$.
We let $\Delta_{\gamma}$, the neighborhood of $\gamma$, be the set of vertices of $\bbZ^2$ located either at a (Euclidean) distance $1/2$ from 
$\tilde \gamma$ 
(when considered as a subset of $\bbR^2$) 
or at a distance $1/\sqrt{2}$ from the meeting point of two non-linked edges in $\gamma$.
We split the $\Delta_{\gamma}$ into two disjoint sets, the internal and the external neighborhoods of $\gamma$  (see Figure \ref{compa})
$$\Delta^-_{\gamma}:=\Delta_{\gamma}\cap \bar \gamma \quad \text{ and } \quad  \Delta^+_{\gamma}:=\Delta_{\gamma}\cap \bar \gamma^{\cc}.$$
We let $\cC$ denote the set contours in $\bbZ^2$ and $\cC_{\gL}$ that of signed contours such that $\bar \gamma \subset \gL$.

 \begin{figure}[ht]
\begin{center}
\leavevmode
\epsfysize = 3 cm
\epsfbox{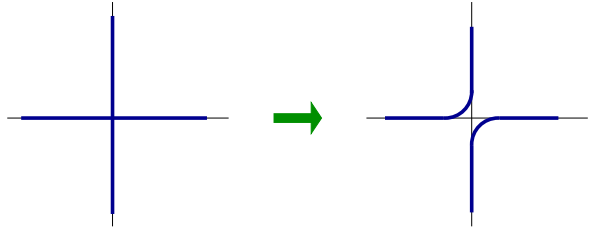}
\end{center}
\caption{\label{linked} 
The rule for splitting a four edges meeting at one points into two pairs of linked edges.
To obtain the set of contours that separates $\{x \ : \ \phi(x)\ge h\}$ from   $\{x \ : \ \phi(x)< h\}$ for $h\in \bbZ$,
we draw all dual edges separating two sites $x$, $y$ such that $\phi(x)\ge h>\phi(y)$ and apply the above graphic rule for every dual vertex where four edges meet.
When several sets of level lines include the same contour, it corresponds to a cylinder of intensity $2$ or more for $\phi$.}
\end{figure}

\medskip

Given $\phi\in \gO_{\gL}$, we say that $\gamma \in \cC_{\gL}$  is a contour for $\phi$ if (using the convention that $\phi(y)=0$ if $y\notin \gL$)

\begin{equation}\label{cont+}
\min_{x\in \Delta^-_{\gamma}, y\in \Delta^+_{\gamma}} \gep(\gamma)(\phi(x)-\phi(y))>0.
\end{equation}
The positive quantity
\begin{equation}
\label{lintens}
k(\gamma,\phi):= \min_{x\in \Delta^-_{\gamma}, y\in \Delta^+_{\gamma}} \gep(\gamma)(\phi(x)-\phi(y))
\end{equation}
 is called the \textit{intensity} of the contour and the
triplet $(\gamma,n)=(\tilde \gamma,\gep(\gamma),k)$ is called a \textit{cylinder}. We say that  $(\gamma,k)$
is a cylinder for $\phi$ if $\gamma$ is a contour of intensity $k$.
The cylinder function associated to $(\gamma,k)$ is defined by
\begin{equation}
\varphi_{(\gamma,k)}=\gep(\gamma) k \ind_{\bar \gamma}.
\end{equation}
We denote by $\hat\cC_{\gL}$ the set of cylinders in $\gL$.
We let $\hat \Upsilon(\phi)$ denote the set of cylinders for $\phi$ and $\Upsilon(\phi)$ the associated set of signed contours.

\medskip

We say that a finite subset $\gL$ of $\bbZ^2$ is \textit{simply connected} if it can be written in the form $\bar \gamma_{\gL}$ for some contour $\gamma_{\gL}$.
We note that for a simply connected set $\phi\in \gO_{\gL}$ is uniquely characterized by its cylinders. More precisely, we have  
\begin{equation}\label{cylinder}
\phi(x):= \sum_{(\gamma,k) \in \hat \Upsilon(\phi)}  \varphi_{(\gamma,k)}.
\end{equation}
The assumption of simple connectedness is present to ensure that there are no level lines in $\phi$ that surrounds holes in $\gL$ (and thus would not be contours).
Furthermore, the reader can check that
\begin{equation}\label{express}
 \cH_{\gL}(\phi)=\sum_{(\gamma,k) \in \hat \Upsilon(\phi)} k|\tilde \gamma|.
\end{equation}
Of course not every set of cylinder is of the form $\hat \Upsilon(\phi)$ and we must introduce a notion of compatibility
which characterizes the ``right" sets of cylinder.

\medskip

Two cylinders $(\gamma,k)$ and  $(\gamma',k')$ are said to be compatible if they are
cylinders for the function $\varphi_{(\gamma,\gep, n)}+\varphi_{(\gamma',\gep', n')}$.
This is equivalent to the three following conditions  being satisfied (see Figure \ref{compa}):
\begin{itemize}
 \item [(i)] $\tilde \gamma\ne \tilde \gamma'$ and $\bar \gamma' \cap  \bar \gamma\in\{\emptyset, \bar \gamma', \bar \gamma\}$.
  \item [(ii)]  If $\gep(\gamma)= \gep(\gamma')$ and $\bar \gamma' \cap \bar \gamma= \emptyset$, then
 then $\bar \gamma'\cap \Delta^+_{\gamma}=\emptyset$ .
 \item [(iii)] If $\gep(\gamma)\ne \gep(\gamma')$ and $\bar \gamma'\subset \bar \gamma$ (resp.\ $\bar \gamma\subset \bar \gamma'$) then 
 $\bar \gamma'\cap \Delta^-_{\gamma}=\emptyset$  (resp. $\bar \gamma\cap \Delta^-_{\gamma'}=\emptyset$).

\end{itemize}
This first condition simply states that compatible contours do not cross each-other.
The conditions $\bar \gamma'\cap \Delta^+_{\gamma}=\emptyset$ and  $\bar \gamma'\cap \Delta^-_{\gamma}=\emptyset$  in  $(ii)$ and $(iii)$  
can be reformulated as: $\gamma$ and $\gamma'$ do not share edges, and if both $\gamma$ and $\gamma'$ possesses two edges adjacent to one vertex 
$x^*\in (\bbZ^2)^*$ 
 then the two edges in $\gamma$ are linked and so are those in $\gamma'$.

 \begin{figure}[ht]
\begin{center}
\leavevmode
\epsfxsize =5 cm
\epsfbox{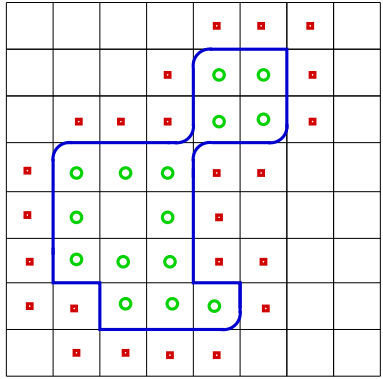}
\end{center}
\caption{\label{compa} 
A contour $\gamma$ represented with its internal (circles) and external (squares) neighborhood.
To be compatible with $\gamma$, a contour $\gamma'$ of the same sign such that $\bar \gamma'\cap \bar \gamma= \emptyset$ cannot enclose any squares.
A compatible contour of opposite sign enclosed in $\gamma$ (such that $\bar \gamma'\subset \bar \gamma$) cannot enclose any circles.
}
\end{figure}

Note the compatibility of two cylinders does not depend on their respective intensity, so that the notion can naturally be extended to  contours:
 The contours  $\gamma$ and $(\gamma'$ are said to be compatible if the cylinders $(\gamma,1)$ and $(\gamma',1)$ are.

\medskip

A collection of cylinders (or of signed contours) 
is said to be a compatible collection if its elements are pairwise compatible.
The reader can check by inspection that the following result holds. In 
particular it establishes that the set of compatible collections of cylinder is in bijection with $\gO_{\gL}$.

\begin{lemma}
If $\gL$ is simply connected, then
for any $\phi\in \gO_{\gL}$, $\hat \Upsilon(\phi)$ is a compatible collection of cylinders
and reciprocally, if $\hat \cA \subset \hat \cC_{\gL}$ is a compatible collection of cylinder in $\gL$
then its elements are the cylinders of the function 
$\sum_{\hat \gamma\in \hat \cA} \varphi_{\hat \gamma}.$
\end{lemma}
 
 \begin{rem}
 The above result does not necessary hold when $\gL$ is not a simply connected set, as in that case, 
the set of level lines enclosing holes in $\gL$ need to satisfy some additional requirement beyond compatibility in order to fit the boundary condition.
 \end{rem}
 
\noindent This description of $\phi$ in terms of compatible collections of cylinders gives a convenient description of the measure $\bP_{\gL,\gb}$.

\begin{lemma} \label{geom}
Conditioned to the set of contours  $\Upsilon(\phi)$, 
the intensities of the corresponding cylinders  under $\bP_{\gL,\gb}$ are independent geometric random variables.

\medskip

More precisely if
 $\hat A \in \hat \cC_{\gL}$ is a compatible collection of cylinders and $\cA\subset \cC_{\gL}$ is the corresponding 
 collection of signed contours then 

\begin{equation}
\bP_{\gL,\gb}\left[ \hat \Upsilon(\phi)=\hat \cA \ | \   \Upsilon(\phi)=\cA \right]
= \prod_{(\gamma,k)\in \hat \cA} (1-e^{-\gb k |\tilde \gamma|}) e^{-\gb k |\tilde \gamma|}.
 \end{equation}
\end{lemma}
\begin{proof}
 This is a direct consequence of \eqref{express} and of the definition of $\bP_{\gL,\gb}$.
\end{proof}

To complete the description of the distribution of  $\phi$ in terms of level lines, we must give a description of the distribution of $\Upsilon(\phi)$.
We let $\bQ_{\gL,\gb}$ be the distribution of a random element $\chi\in \cP(\cC_{\gL})$ (the set of parts of $\cC_{\gL}$) under which
the variables $\chi(\gamma):=\ind_{\{\gamma\in \chi\}}$ are independent and 
\begin{equation}\label{defgroq}
\bQ_{\gL,\gb}\left( \gamma\in \chi \right)=e^{-\gb|\tilde \gamma|}. 
\end{equation}

\begin{lemma}\label{restrict}
We have 
\begin{equation}\label{condit}
 \bP_{\gL,\gb}[\hat \cC_{\phi} \in \cdot]:= \bQ_{\gL,\gb}\left[ \chi \in \cdot \ | \  \chi \text{ is a compatible collection } \right].
\end{equation}
In particular the distribution of $\hat\cC_{\phi}$ is stochastically dominated (for the inclusion) by $\hat \bP_{\gL,\gb}$.
 \end{lemma}

\begin{proof}
Equation \eqref{condit}  can be deduced from \eqref{express}:
After summing over the intensities, we obtain that the probability of observing a given collection of signed contour 
$\cA$ is proportional to
$$ \prod_{(\gamma,\gep)\in \cA} \sum_{k\ge 1} e^{-\gb k |\tilde \gamma|}=\frac{1}{ e^{\gb|\tilde \gamma|}-1}.$$
This corresponds to the relative weights of a random the set of contour is given by independent Bernouilli variables with parameter satisfying 
$p_{\gamma}/(1-p_{\gamma})= (e^{\gb|\tilde \gamma|}-1)^{-1}.$

\medskip

We observe that $\{ \chi \text{ is a compatible collection}\}$ is a decreasing event for the inclusion. 
Hence the stochastic domination follows immediately from the FKG-Harris inequality (recall \eqref{FKG} and the paragraph below it).

\end{proof}

\subsection{Exponential decay  and asymptotics for peak probabilities} \label{Peierls}

The probability for extremal events can be estimated by evaluating the probability for the corresponding contour realization.
This type of argument dates back to \cite{cf:Pei}.
The result we present here is an extension of  \cite[Proposition 3.9]{cf:PLMST2} which only treats the case of single peaks. 
The proof of our generalization uses the same ideas, and appears  in Appendix \ref{ppeaks} for the sake of completeness. 
\begin{proposition}\label{rourou}
For any $\gb$ satisfying condition \eqref{finite}, there exists a constant $C(\gb)$ such that  for any $n\ge 1$ any $\gL$ and any triple of distinct vertices $(x,y,z)$
such that $x\sim y \sim z$ we have 
\begin{equation}
 \begin{split}
 \frac{1}{2}e^{-4\gb n}&\le  \bP_{\gL,\gb}[\phi(x)\ge n]\le C e^{-4\gb n},\\
  \frac{1}{4}e^{-6\gb n}& \le \bP_{\gL,\gb}[\min(\phi(x),\phi(y)) \ge n]\le C e^{-6\gb n},\\
   \frac{1}{8}e^{-8\gb n}&\le  \bP_{\gL,\gb}[\min(\phi(x),\phi(y),\phi(z))\ge n]\le C e^{-8\gb n}.
 \end{split}
\end{equation}
\end{proposition}

Using Theorem \ref{infinitevol} we deduce from Proposition \ref{rourou} the exact asymptotic behavior for the probability of have high peaks of size one and two
(recall the definitions \eqref{limzits} and \eqref{defq}). The details of the proof are also given in Appendix \ref{ppeaks}.

\begin{proposition}\label{sendingout}
 For $\gb>\gb_0$ (given by Theorem \ref{infinitevol})  there exist positive constants $\alpha_1$, $\alpha_2$,
 $c$ and $C$  such that 
 for any $n\ge 1$, $\gL\subset \bbZ$ and $x,y\in \in \gL$ with $x\sim y$ we have 
 \begin{equation}\label{qtronz}
 \begin{split}
  \left|\bP_{\gL,\gb}[\phi(x)\ge n]- \alpha_1 J^{2n} \right|&\le C \left(J^{3n}+ 
  e^{-c  d(x,\partial\gL)} \right),\\
  \left|\bP_{\gL,\gb}[\min(\phi(x),\phi(y))\ge n]- \alpha_2 J^{3n}\right|&\le C \left( nJ^{4n}+ 
  e^{-c  d(x,\partial\gL)} \right),
  \end{split}
 \end{equation}
where $d(x,\partial\gL)$ is defined in  \eqref{disset}. 
We also have (recall \eqref{defq})
 \begin{equation}\label{qtronzac}
 \begin{split}
  \left|\bP_{\gL,\gb}[q(\phi,x,n)=1]- \alpha_1 J^{2n} \right|&\le C \left(J^{3n}+ 
  e^{-c  d(x,\partial\gL)} \right),\\
  \left|\bP_{\gL,\gb}[q(\phi,x,n)=2]- 4\alpha_2 J^{3n}\right|&\le C \left( nJ^{4n}+ 
  e^{-c  d(x,\partial\gL)} \right).
  \end{split}
 \end{equation}
 Both results remain  valid for $\bP_{\gb}$ (only the first error term being kept in that case).
\end{proposition}

\section{Identification of the critical point and rough lower bound} \label{eloweb}

In this section, we prove a lower bound on the free energy which has the right order of magnitude, but is not sharp up to constant.
In particular this lower bound combined with the observation made below \eqref{youpee}, allows to 
identify the value of $h_w(\gb)$.

\begin{proposition}\label{activiax}
Assume that $\gb$ satisfies condition \eqref{finite}.
Then we have $$h_w(\gb)=\log \left(\frac{e^{4\gb}}{e^{4\gb}-1}\right)$$ and 
there exists a constant $c_{\gb}>0$ such that for all $u\in (0,1]$, we have  
\begin{equation}\label{ordretroi}
\bar \tf(\gb,u)\ge c_{\gb} u^3.
\end{equation}
\end{proposition}

The section is organized as follows:
In Section \ref{recodex} we turn into a rigorous statement the claim made in Section \ref{houra} about the change of boundary condition,
and we use it in Section \ref{izigo} to prove Proposition \ref{activiax}.
The proof of the sharper bound, developed in Section \ref{nozigo} also builds on the ideas exposed in the present section.

\subsection{Changing the boundary condition and rewriting the partition function}\label{recodex}

We are going to work with the following alternative characterization of the free energy.

\begin{lemma}\label{ojko}
 
 For any $n\ge 1$, we have 
 \begin{equation}\label{lefromage}
 \bar \tf(\gb,u)=\lim_{N\to \infty} \frac{1}{N^2} \log \bE_{N,\gb} \left[ e^{u |\phi^{-1}[n,\infty)|- \bar H(\phi^{-1}[n,\infty))} \right].
 \end{equation}
\end{lemma}

\begin{proof}

Recalling the definition \eqref{defhamil}, we have for every $\phi\in \gO_N$ 
$$ |\cH^n_N(\phi)-\cH_N(\phi)|\le  4nN.$$ 
As a consequence we have 
$$e^{-4nN\gb}\le \frac{\dd \bP^n_{N,\gb}}{\dd \bP_{N,\gb}}(\phi)\le e^{4nN\gb},$$
and thus recalling \eqref{youpee} we obtain that 
\begin{equation} \label{lacomparade}
  \left| \log \frac{\cZ^{u-\log(1-J^2)}_{N,\gb}}{\cZ_{N,\gb}}- \log \bE^n_{N,\gb} \left[ e^{u |\phi^{-1}(\bbZ_{-})|- \bar H(\phi^{-1}(\bbZ_{-}))} \right] \right| \le
  4nN \gb,
\end{equation}
Finally, noticing that by symmetry $\bP^n_{N,\gb}[ n-\phi \in \cdot]=\bP_{N,\gb}[ \phi \in \cdot]$ we have
\begin{equation}
\bE^n_{N,\gb} \left[ e^{ u |\phi^{-1}(\bbZ_{-})|- \bar H(\phi^{-1}(\bbZ_{-}))} \right]
 =\bE_{N,\gb} \left[ e^{u |\phi^{-1}[n,\infty)|- \bar H(\phi^{-1}[n,\infty)} \right],
\end{equation}
and thus 
\begin{multline}
 \lim_{N\to \infty} \frac{1}{N^2} \log \bE_{N,\gb} \left[ e^{u |\phi^{-1}[n,\infty)|- \bar H(\phi^{-1}[n,\infty))} \right]\\
 =  \lim_{N\to \infty} \frac{1}{N^2}  \log \frac{\cZ^{u-\log(1-J^2)}_{N,\gb}}{\cZ_{N,\gb}}=\bar \tf(\gb,u).
 \end{multline}
 \end{proof}

\subsection{Proof of Proposition \ref{activiax}} \label{izigo}

\noindent Note that Lemma \ref{ojko} implies, by Jensen's inequality, that for every $n$ 
\begin{equation}\label{jensenox}
  \bar \tf(\gb,u)\ge \limsup_{N\to \infty} \frac{1}{N^2}\bE_{N,\gb} \big[ u |\phi^{-1}[n,\infty)|- \bar H(\phi^{-1}[n,\infty)) \big].
\end{equation}
To prove a lower bound on    $\bar \tf$, we can replace the term  $\bar H(\phi^{-1}[n,\infty))$ by an effective lower bound.
Recall the definition \eqref{defq} and set 
for any $i\ge 1$
\begin{equation}
\begin{split}
f^n_i(\phi):=\#\{ x\in \gL_N \ : \ q(\phi,x,n)=i \}, \\
f^n_{i+}(\phi):=\#\{ x\in \gL_N \ : \ q(\phi,x,n)\ge i \}.
\end{split}
\end{equation}
Using Lemma \ref{lehagga}, we have for any $u\ge 0$  
\begin{equation}\label{daabaa}
 u |\phi^{-1}[n,\infty)|- \bar H(\phi^{-1}[n,\infty))\ge u f^n_{1+}(\phi)- c_2(\gb) f^n_{2+}(\phi).
\end{equation}
Thus from \eqref{jensenox}, we obtain
\begin{equation}\label{jaxx}
\bar \tf(\gb,h)\ge \limsup_{N\to \infty}\frac{1}{N^2} \bE_{N,\gb}\left[u f^n_{1+}(\phi)- c_2(\gb) f^n_{2+}(\phi) \right].
\end{equation}
Using Proposition \ref{rourou} we immediately obtain
\begin{equation}\label{zeone}
\frac{1}{N^2} \bE_{N,\gb}\left[f^n_{1+}(\phi)\right]\ge \inf_{x\in \gL_N}  \bP_{N,\gb}[\phi(x)\ge n] \ge \frac{1}{2}J^{2n}.
\end{equation}
Now $f^n_{2+}(\phi)$ can be bounded above by twice the number of couple of neighboring points above level $n$.
As there are $2(N-1)N$ edges in $\gL_N$ we obtain that using Proposition \ref{rourou} that
\begin{equation}\label{zetoo}
\frac{1}{N^2} \bE_{N,\gb}\left[f^n_{2+}(\phi)\right]\le  \frac{4N(N-1)}{N^2}
\sup_{\{x,y\in \gL_n \ : \ x\sim y\}}  \!\!\!\!\!  \bP_{N,\gb}[\min(\phi(x),\phi(y))\ge n] \ge  C J^{3n}.
\end{equation}
Thus combining \eqref{jaxx},\eqref{zeone} and \eqref{zetoo}, and changing the value of $C$, we obtain
\begin{equation}
 \bar \tf(\gb,h)\ge \frac{1}{2}J^{2n}u- C J^{3n},
\end{equation}
and as $n\ge 1$ is arbitrary the result follows after optimizing over $n$ (that is, we must choose $J^n$ of order $u$ times a small constant).
\qed

 \section{The sharp lower bound on the free energy}\label{nozigo}
 
We explain now how to improve the method developed in the previous section in order to obtain an optimal upper bound.
 
 \begin{proposition}\label{activriax}
For $\gb>\gb_0$ (of Theorem \ref{infinitevol}), there exists  a constant $C$ such that for every $\gb\ge \gb_0$, for every $u\in(0,1]$
\begin{equation}\label{ordretdroi}
\bar \tf(\gb,h)\ge F(\gb,u)- C |\log u| u^{24/7}
\end{equation}
\end{proposition}

We are not going to prove \eqref{ordretdroi} directly, but rather prove that there exists a constant $C$ such that 
 for every $\gb\ge \gb_0$, for every $u\in(0,1]$ and $n\ge 1$
\begin{equation}\label{gimik}
 \bar \tf(\gb,u)\ge \alpha_1 J^{2n} -\frac{2\alpha_2(J^3-J^4)}{1-J^4} J^{3n}-C \left(uJ^{20n/7}+ n J^{24n/7} \right). 
\end{equation}
Choosing $n$ that maximizes  $\alpha_1 J^{2n} -\frac{2\alpha_2(J^3-J^4)}{1-J^4} J^{3n}$
we can check that the last term is of order $|\log u| u^{24/7}$. Thus \eqref{gimik} implies \eqref{ordretdroi}.
 
 \subsection{An inconclusive attempt}
 
 Before going to the proof of Proposition \ref{activriax}, let us shortly discuss 
 some points in the previous proof where we can tighten the estimates.
 Firstly we know exactly the penalty induced by clusters of size $2$ and thus it would seem more reasonable, instead of \eqref{daabaa} to use the sharper estimate
\begin{equation}\label{scroos}
 u |\phi^{-1}[n,\infty)|- \bar H(\phi^{-1}[n,\infty))\ge u f^n_{1+}(\phi)- \frac{1}{2}
 \log\left( \frac{1-J^4}{1-J^3}\right) f^n_{2}(\phi) -c_{2} f^n_{3+}(\phi).
\end{equation}
Secondly, to estimate $ N^{-2}\bE_{N,\gb}\left[f^n_{1+}(\phi)\right]$ and $N^{-2}\bE_{N,\gb}\left[f^n_{2}(\phi)\right]$, it seems more appropriate to use 
Proposition \ref{sendingout} instead of Proposition \ref{rourou}.
For large values of $n$, the two quantity can respectively be approximated by $\alpha_1 J^{2n}$ and $4\alpha_2 J^{3n}$.
These changes would yield a bound of the type 
\begin{equation}\label{ratherbad}
\bar \tf(\gb,u)\ge \alpha_1 J^{2n}u - 2 \alpha_2  \log\left( \frac{1-J^4}{1-J^3}\right)J^{3n}- R(n,u),
\end{equation}
where $R(n,u)$ is of a smaller order than two first terms.
This is worse than the desired \eqref{gimik}.

\medskip

The reason why this bound is not sharp  is that  Jensen's inequality in \eqref{jensenox} is slightly suboptimal.
Indeed it can be checked that applying Jensen's inequality in \eqref{wouhou}, 
we would get something similar to the r.h.s.\ of \eqref{ratherbad}.

\medskip

In order to cope with the problem, we must use the fact peaks of size two (the term causing trouble is the one involving $f^n_2$)
are distributed in a very uncorrelated fashion.
We achieve this in a rather indirect manner, by looking at the conditional expectation of the partition function with respect to an intermediate level set of $\phi$ between $0$ and $n$.

\subsection{A second rewriting of the partition function}\label{rewrite}

We decide to push the transformation of the partition function performed in Section \ref{recodex} one step further.
We want to define an auxiliary model in which the Gibbs weight of a trajectory $\phi$ depends on the set of points above level $n-k$ (for $k\in\lint 1,n \rint$) 
but whose partition function is the same as the one appearing in the l.h.s.\ of \eqref{lefromage}.
The reader can check from the definition of $\bP_{N,\gb}$ \eqref{defSOS} that conditionally on $\phi^{-1}[n-k,\infty)= \gG$, 
the restricted process $[\phi-(n-k)]\restrict_{\gG}$ is independent of 
of $\phi\restrict_{\gG^{\cc}}$ and its conditional distribution is given by $\bP^+_{\gG}$ defined on $\gO^+_{\gG}$ by (recall \eqref{ZZplus}
\begin{equation}\label{defpus}
\bP^+_{\gG}(\phi):=\frac{1}{\cZ_{\gG,+}}e^{-\gb\cH_{\gG}(\phi)}.
\end{equation}
Thus we have 
\begin{multline}\label{cruize}
 \bE_{N,\gb} \left[ e^{ u |\phi^{-1}[n,\infty)|- \bar H(\phi^{-1}[n,\infty))} \ | \ \phi^{-1}[n-k,\infty)= \gG \right]\\
  = \bE^+_{\gG} \left[ e^{ u |\phi^{-1}[k,\infty)|-  \bar H(\phi^{-1}[k,\infty))}\right].
\end{multline}
We set 
 $$G^{k,u}(\gG):= \log  \bE^+_{\gG} \left[e^{ u |\phi^{-1}[k,\infty)|-  \bar H(\phi^{-1}[k,\infty))}\right],$$ 
and considering the expectation of \eqref{cruize}  we obtain
 \begin{equation}
 \bE_{N,\gb} \left[ e^{ u |\phi^{-1}[n,\infty)|- \bar H(\phi^{-1}[n,\infty))} \right]=\bE_{N,\gb}\left[ e^{G^{k,u}(\phi^{-1}[n-k,\infty))} \right].
\end{equation}
In particular,
\eqref{lefromage} implies that 
\begin{equation}\label{lefromagess}
 \bar \tf(\gb,u)=\lim_{N\to \infty} \frac{1}{N^2} \log \bE_{N,\gb} \left[ e^{G^{k,u}(\phi^{-1}[n-k,\infty))} \right].
\end{equation}
As for $H$, the value of  $G^{k,u}(\phi^{-1}[n-k,\infty))$ can be recovered by summing over all maximal connected components of $\phi^{-1}[n-k,\infty)$.
Similarly to what was done in Lemma \ref{lehagga}, we want to have a sharp estimate for the contribution for clusters of size $1$ and $2$, and a reasonable control 
over the rest.

\medskip

 We do not need an upper bound on $G^{k,u}\{x,y\}$ to prove Proposition \ref{activriax} but we include one here as it is going to be useful later on to obtain the upper bound on the free energy.
Note that the bound which we obtain in $(iii)$ below is far from optimal  
but it is sufficient for our purpose.

 \begin{lemma}\label{stimaG}

\begin{itemize}
 \item [(i)] For any $u\in \bbR$ we have 
 $$g^k_1(u):=G^{k,u}\{x\}= \log \left(1+J^{2k}(e^u-1)\right).$$
 \item [(ii)] For $x\sim y$, we have  
  $$g^k_2(0):=G^{k,0}\{x,y\}= \log \left(1-\left(\frac{J^3-J^4}{1-J^4}\right)J^{3k}\right).$$
 Moreover, there exists $C$ such that for all $u\in (0,1]$ and all $k$
  $$g^k_2(u)\le g^k_2(0)+CJ^{3k}u.$$
  \item [(iii)] There exists a constant $c_2(\gb)$ such that for any connected set $\gG$ with $|\gG|\ge 3$ we have 
  $$G^{k,0}(\gG)\ge -c_2(\gb) |\gG|.$$ 
 \end{itemize}
 \end{lemma}

 \begin{proof}
Let us rewrite  $e^{G^{k,u}(\gG)}$ in a more convenient fashion.
We have 
\begin{equation}\label{muzak}
 e^{G^{k,u}(\gG)}=\frac{1}{\cZ^{+}_{\gG,\gb}}\sum_{x\in \gO^+_{\gG}} e^{-\gb\cH(\phi)+u|\phi^{-1}[k,\infty)|+\bar H(\phi^{-1}[k,\infty))}.
\end{equation}
To prove $(i)$, we observe that the numerator of \eqref{muzak} is equal to 
 \begin{equation}
 \sum_{i=0}^{k-1} J^{2i}+\sum_{i=k}^{\infty}e^u J^{2i}= \frac{1}{1-J^2}(1+J^{2k}e^u),
 \end{equation}
and thus we can conclude using Lemma \ref{lehagga} to estimate the denominator.

\medskip

To prove $(ii)$, decomposing the sum according to the possible values for $a:=\min(\phi(x),\phi(y))$ and $b:=|\phi(x)-\phi(y)|$,
we obtain for the numerator
\begin{multline}\label{danum}
\cZ^{+}_{\{x,y\}}  e^{G^{k,u}(\gG)}=
\sum_{a=0}^{k-1} J^{3a} \left(1+2\sum_{b=1}^{k-a-1}J^{2b}+ 2\sum_{b=k-a}^{\infty}e^{u}J^{2b}\right)\\
+
\sum_{a=k}^{\infty}e^{2u}\left(\frac{1-J^3}{1-J^4}\right)J^{3a} \left(1+\sum_{b=1}^{\infty}J^{2b} \right).
\end{multline}
When $u=0$ this yields
\begin{multline}
\cZ^{+}_{\{x,y\}}  e^{G^{k,0}(\gG)}=\frac{(1-J^{3k})(1+J^2)}{(1-J^3)(1-J^2)}+\frac{J^{3k}}{(1-J^2)^2}\\
  = \frac{1-J^4}{(1-J^2)^2(1-J^3)}\left[1-\left(\frac{J^3-J^4}{1-J^4}\right)J^{3k}\right].
\end{multline}
Using Lemma \ref{lehagga} $(ii)$ for $\cZ^{+}_{\{x,y\},\gb}$, this yields the desired value for $G^{k,0}\{x,y\}$.
To obtain the upper-bound for positive $u$ we simply notice that from \eqref{danum}
\begin{equation}
\cZ^{+}_{\{x,y\}}\left(e^{G^{k,u}(\gG)}- e^{G^{k,0}(\gG)}\right)=\sum_{a=0}^{k-1} J^{2k+a}(e^u-1)
+\frac{J^{3k}(e^{2u}-1)}{(1-J^2)^2}
   \le C J^{3k} u.
\end{equation}
The third point is a trivial consequence of the upper-bound of Lemma \ref{lehagga} $(iii)$ since the numerator of \eqref{muzak} is larger than $1$.
\end{proof}

\subsection{Proof of Proposition \ref{activiax}}

Fix $u\in(0,1]$. Using Jensen inequality in \eqref{lefromagess}, we obtain
\begin{equation}
 \bar \tf(\gb,u)\ge \limsup_{N\to \infty}\frac{1}{N^2}\bE_{N,\gb}[G^{k,u}[n-k,\infty)].
\end{equation}
We can use Lemma \ref{stimaG} and the fact that $G^{k,u}$ is increasing in $u$
to obtain (similarly to \eqref{scroos}) that 
\begin{equation}
G^{k,u}[n-k,\infty)\ge g^k_1(u) f^{n-k}_1(\phi)
+ \frac{1}{2} g^k_2(0)  f^{n-k}_2(\phi)- c_2(\gb)f^{n-k}_{3+}(\phi).
\end{equation}
Using the fact that $g^k_1(u)\ge u$ we thus  have
\begin{multline}\label{dazkoll}
  \bar \tf(\gb,u)\ge J^{2k} u \lim_{N\to \infty} N^{-2}  \bE_{N,\gb}\left[f^{n-k}_1(\phi)\right]\\
  +  \frac{1}{2}\log \left(1-\left(\frac{J^3-J^4}{1-J^4}\right)J^{3k}\right)\lim_{N\to \infty}N^{-2}\bE_{N,\gb}\left[f^{n-k}_2(\phi)\right]\\
  - c_2(\gb)\lim_{N\to \infty} N^{-2}  \bE_{N,\gb}\left[f^{n-k}_{3+}(\phi)\right],
\end{multline}
provided the limits exists. This existence is guaranteed by the following result (which is an immediate consequence of \eqref{decayz}).

\begin{lemma}
  For any $i$ and $m$ we have 
 \begin{equation}\label{ispoor}\begin{split}
  \lim_{N\to \infty} N^{-2}  \bE_{N,\gb}\left[f^{m}_{i}(\phi)\right]=\bE_{\gb}[q(\phi,{\bf 0},m)=i],\\
  \lim_{N\to \infty} N^{-2}  \bE_{N,\gb}\left[f^{m}_{i+}(\phi)\right]=\bE_{\gb}[q(\phi,{\bf 0},m)\ge i].
  \end{split}
 \end{equation}
\end{lemma}

\noindent The next step is to replace $\bE_{\gb}[q(\phi,{\bf 0},m)=i]$ and $\bE_{\gb}[q(\phi,{\bf 0},m)\ge i]$ by their asymptotic approximation.
Recalling Proposition \ref{rourou} and \ref{sendingout} we have
\begin{equation}
\begin{split}
 |\bP_{\gb}[q(\phi,{\bf 0},m)=1]- \alpha_1 J^{2m}|&\le C J^{3m},\\
 |\bP_{\gb}[q(\phi,{\bf 0},m)=2]- 4\alpha_2 J^{3m}|&\le C m J^{4m},\\
 |\bP_{\gb}[q(\phi,{\bf 0},m)\ge 3]|&\le C m J^{4m}.\\
 \end{split}
\end{equation}
We also observe (using a Taylor expansion) that 
\begin{equation}\label{isrich}
\left|g^k_2(0)+\frac{J^3-J^4}{1-J^4}J^{3k}\right|\le C J^{6k}. 
\end{equation}
Thus, using Equations \eqref{ispoor}-\eqref{isrich} in \eqref{dazkoll} we obtain (for some constant $C$ which depends on $\gb$)
\begin{multline}
   \bar \tf(\gb,u)\ge  J^{2k}\left(\alpha_1 J^{2(n-k)}- C J^{3(n-k)}\right) u \\
   - 
   \frac{1}{2} \left(\left(\frac{J^3-J^4}{1-J^4}\right)J^{3k}+C J^{6k} \right)\left(4\alpha_2 J^{3(n-k)}+C (n-k)J^{4(n-k)} \right)-C (n-k) J^{4(n-k)}\\
   \le \alpha_1 J^{2n} u-\frac{2\alpha_2(J^3-J^4)}{1-J^4}J^{3n}-C \left( J^{3n-k}u +  J^{3(n+k)}+ (n-k) J^{4(n-k)} \right).
\end{multline}
To conclude the proof of \eqref{gimik} we just need to choose $k=n/7$.

\section{An easier upper bound on the free energy}\label{upeb}

As we did for the lower bound, we  first present a proof of a rougher upper bound for the free energy which gives  the right order of magnitude (at least when \eqref{finite} holds) 
and that is valid for all values of $\gb$.
It is an immediate consequence of Lemma \ref{uppg} below which is also of crucial importance to prove the sharp upper bound.

\begin{proposition}\label{easyup}
 For every $\gb>0$
 there exists a constant $C$ such that for every $u\in [0,1]$
 \begin{equation}
  \tf(\gb,h_c(\gb)+u)\le C u^3.
 \end{equation}

\end{proposition}

\subsection{Estimating the contribution of large clusters}

To prove upper bound results, we use the alternative expression for the free energy introduced in \eqref{lefromagess}.
We can rely on Lemma \ref{stimaG} to estimate the contribution to $G^{k,u}(\phi^{-1}[n-k,\infty))$ of clusters of size one and two,
but we yet need some upper bound for clusters of larger size.
This is the purpose of the following result.

\begin{lemma}\label{uppg}
Given $\gb>0$, there exist constants $K>0$  and $\tilde c_1(\gb)>0$ such that for all $k \in \left[1,\frac{|\log u|}{2\gb}-K\right]$, and all connected sets $\gG$ satisfying $|\gG|\ge 2$, we have
\begin{equation}\label{tii}
 G^{h,k}(\gG)\le -\tilde c_1(\gb) J^{3k}.
\end{equation}
\end{lemma}

\noindent The above result implies in particular that for  $k \in \left[1,\frac{| \log u|}{2\gb}-K\right]$, the contribution of clusters of size two and larger is negative and thus that 
\begin{equation}\label{simpel}
 G^{k,u}(\phi^{-1}[n-k,\infty))\le g^k_1(u) f^{n-k}_1(\phi).
\end{equation}
We already obtain a bound of the right order on the free energy only by using this information for the largest value allowed for $k$.

\begin{proof}[Proof of Proposition \ref{easyup}]
 For $u$ sufficiently small 
let us fix $k= \left\lfloor \frac{|\log u|}{2\gb}-K \right\rfloor$ and $n> k$ arbitrary. Using \eqref{simpel} we have
 \begin{equation}
\bE_{N,\gb} \left[ e^{G^{k,u}(\phi^{-1}[n-k,\infty))}\right]
  \le \bE_{N,\gb}  \left[ e^{a^{k}_1(u) f^{n-k}_1(\phi)}\right]
 \le  e^{N^2 a^k_1(u)}.
 \end{equation}
We can then conclude using \eqref{lefromagess} and expression given in Lemma \ref{stimaG} for $a^k_1(u)$ that 
$$\bar \tf(\gb,u)\le a^k_1(u)\le 2 u J^{2k}$$
provided that $k$ is sufficiently large. To conclude we just need to replace $k$ by its value in the above expression.
\end{proof}

\subsection{Proof of Lemma \ref{uppg}}

We can assume that $u$ is small as if not, the interval\\
 $\left[1,\frac{|\log u|}{2\gb}-K\right]$ could be  empty.
Using Lemma \ref{lehagga}, we have (recall \eqref{defpus})
\begin{equation}\label{spoof}
e^{G^{k,u}(\gG)}= \bE^+_{\gG}\left[e^{u\phi^{-1}[k,\infty)+\bar H\left( \phi^{-1}[k,\infty)\right) }\right]
\le   \bE^+_{\gG}\left[e^{u f^k_{1+}(\phi)- c_1(\gb)f^k_{2+}(\phi)}\right].
  \end{equation}
We split now $f^k_{1+}(\phi)$ and $f^{k}_{2+}(\phi)$ in two parts corresponding to the respective contribution of 
\textit{odd} and \textit{even} sites (we call sites of $\bbZ^2$ odd resp.\ even when the sum of their coordinates are odd resp. even)
\begin{equation}\begin{split}
 f^k_{1+}(\phi)=f^{k,\odd}_{1+}(\phi)+f^{k,\even}_{1+}(\phi),\\
 f^{k}_{2+}(\phi)=f^{k,\odd}_{2+}(\phi)+f^{k,\even}_{2+}(\phi).\\
\end{split}\end{equation}
The Cauchy-Schwartz inequality applied to \eqref{spoof} yields
\begin{equation}\label{CS}
e^{2G^{k,u}(\gG)}\le \bE^+_{\gG}\left[e^{2u f^{k,\odd}_{1+}(\phi)-2c_1(\gb) f^{k,\even}_{2+}(\phi)}\right]
   \bE^+_{\gG}\left[e^{2u f^{k,\even}_{1+}(\phi)- 2c_1(\gb)f^{k,\odd}_{2+}(\phi)}\right].
   \end{equation}
In order to estimate each factor in the r.h.s., we are going to condition to  realization the field on half of the lattice sites and then use
independence.
Let $\gG_{\odd}$ and $\gG_{\even}$ denote the set of odds and even sites in $\gG$ respectively.
We need to prove that the  following holds for some positive constant $c$
 \begin{equation}\begin{split}\label{labonte}
  \bE^+_{\gG}\left[e^{2u f^{k,\odd}_{1+}(\phi)- 2c_1(\gb)f^{k,\even}_{2+}(\phi)} \ | \ \phi \restrict_{\gG_{\odd}} \right]&\le  e^{-c J^{k} f^{k,\odd}_{1+}(\phi)},\\
  \bE^+_{\gG}\left[e^{2u f^{k,\even}_{1+}(\phi)- 2c_1(\gb)f^{k,\odd}_{2+}(\phi)} \ | \ \phi \restrict_{\gG_{\even}} \right]&\le  e^{-c J^{k} f^{k,\even}_{1+}(\phi)}.
  \end{split}
 \end{equation}
It is immediate to check that 
$$\bP^+_{\gG,\gb}\left[ \phi(x)\ge k \ | \ \phi(y), y\ne x \right]\ge J^{2k},$$
meaning that $f^{k}_{1+}(\phi)$ stochastically dominates a sum of $|\gG|$ independent Bernoulli variables of parameter $J^{2k}$.
As a consequence of \eqref{labonte} we obtain thus
\begin{multline}
\bE^+_{\gG}\left[e^{2u f^{k,\odd}_{1+}(\phi)-2c_1(\gb) f^{k,\even}_{2+}(\phi)}\right]\le \bE^+_{\gG}\left[ e^{-c J^{k} f^{k,\odd}_{1+}(\phi)}\right]\\
\le  \left( 1- J^{2k}\left(1-e^{-c J^{k}}\right)\right)^{|\Gamma_{\odd}|}\le e^{-c'|\gG_{\odd}|J^{3k}}.
\end{multline}
The other factor in \eqref{CS} can be bounded in the same manner and this yields \eqref{tii}.

 \medskip

\noindent By symmetry, we only have to prove the first inequality in \eqref{labonte}.
Note that conditioned to $\phi \restrict_{\gG_{\odd}}$, the variables $(\phi(x))_{x\in \gG_{\even}}$ are independent.
Further more we have 
\begin{equation}\label{claim}
   \bE^+_{\gG}\left[ \phi(x)\ge k \ | \ \phi \restrict_{\gG_{\odd}} \right]\ge \frac{J^{k}}{1+J}\ind_{\{\exists y, y\sim x \text{ and } \phi(y)\ge k \}}.
\end{equation}
Indeed, setting $$v_j(\phi):=  \exp\left(-\gb \sum_{y\sim x}|j-\phi(y)|\right),$$ where by convention $\phi(y)$ is taken to be $0$ if $y\in \partial \gG$,
the left-hand side is equal to
\begin{equation}
    \bE^+_{\gG}\left[ \phi(x)\ge k \ | \ \phi \restrict_{\gG_{\odd}} \right]=\frac{\sum_{j\ge k} v_j(\phi)}{\sum_{j\ge 0} v_j(\phi)}.
\end{equation}
The reader can check that the r.h.s.\ is increasing in $\phi$ (this is in fact a consequence of the FKG inequality \eqref{FKG})
and thus it is sufficient 
to check the case where $\phi$ is equal to $k$ for one neighbor and $0$ for the others. In that case we have 
\begin{equation}
 \frac{\sum_{j\ge k} v_j(\phi)}{\sum_{j\ge 0} v_j(\phi)}=\frac{(1-J)J^k}{(1-J^2)(1-J^k)+(1-J)J^k}\le  \frac{J^{k}}{1+J},
\end{equation}
which proves the claim \eqref{claim}.
Now, we observe now that 
\begin{equation}
\#\{ x\in \gG_{\even} \ | \ \exists y\in \gG_{\odd}, \ y\sim x \text{ and } \phi(y)\ge k\} \ge 
f^{k,\odd}_{1+}(\phi)/4.
\end{equation}
As for all the vertices counted above, $\phi(x)\ge k$ implies $q(\phi,x,k)\ge 2$, we have for some adequate choice of $c>0$, we have
\begin{multline}
   \bE^+_{N,\gb}\left[e^{-2c_1(\gb) f^{k,\even}_{2+}(\phi)} \ | \ \phi \restrict_{\gG_{\odd}} \right]\\
   \le \left( 1+(1-e^{-2c_1(\gb)})\frac{J^{k}}{1+J}\right)^{f^{k,1}_{1+}(\phi)/4}\le e^{-c J^{k} f^{k,\odd}_{1+}(\phi)}.
\end{multline}
With our condition on $k$, this yields \eqref{labonte}.

\qed

\section{Sharpening  the upper bound}\label{sub}

We finally prove the sharp upper bound on the free energy,

\begin{proposition}\label{superbien}
There exists $\gep>0$ and $C>0$ such that for any $\gb\ge \gb_0$ (from Theorem \ref{infinitevol}) for all $u\in(0,1]$ we have
\begin{equation}\label{gross}
 \bar\tf(\gb,u)\le F(\gb,u)+C u^{3+\gep}.
\end{equation}
\end{proposition}

\subsection{Decomposition of the proof}

The proof of this statement is the most delicate part of the paper.  
We use the characterization of the free energy provided by \eqref{lefromagess}.
The idea is to factorize $\bE_{N,\gb}\left[e^{G^{k,u}\left(\phi^{-1}[n-k,\infty)\right)}\right]$ by dividing 
the box $\gL_N$ into cells of fixed size. If the size of the cells is
chosen sufficiently large (but depending only on $u$ and not on $N$) 
$\phi$ should be almost independent in different cells because of the fast decay of correlations \eqref{decayz}.

\medskip

To fix ideas we choose the value of the parameters used in the cell decomposition now. Other choices are possible and we do not try to optimize the value of $\gep$ in \eqref{gross}.
Given $u$ we set 
\begin{equation}\label{ramets}
k_u:= \left\lceil \frac{9 |\log u|}{20 \gb} \right\rceil,\quad 
M_u:= \left\lceil u^{-1/4}\right\rceil^2 \quad \text{and} \quad  \ell_u=\left\lceil (\log u)^2 \right \rceil.
\end{equation}
The value of $n$ (in \eqref{lefromagess}) to be chosen is somehow arbitrary, the bound that we prove being independent of it.

\medskip

We split $\gL_N$ into cells of side-length $M=M_u$ assuming that 
$N$ is a multiple of $M$. For $z\in \lint 0, M/N-1 \rint^2$ we set  $$\cB_z:=\gL_M+ Mz.$$
For technical reasons, it turns out to be convenient to define also a smaller cell contained in $\cB_z$.
We set 
$$\bar \cB_z:= \left\lint \sqrt{M}, M-\sqrt{M}\right\rint^2 +Mz.$$

Recall that $\Upsilon(\phi)$ denote the set contours of $\phi$ defined in Section \ref{contour}.
We let  $\Diam(\gamma)$ denote the Euclidean diameter of the geometric contour $\tilde \gamma$ considered as a subset of $\bbR^2$.
Recalling \eqref{ramets} we define
\begin{equation}
\Upsilon^{\larg}(\phi):=\{ \gamma \in \Upsilon(\phi) \ : \ \Diam(\gamma)\ge \ell_u \},
\end{equation}
and let $\hat \Upsilon^{\larg}(\phi)$ denote the associated set of cylinders.
We say that a cell $\cB_z$ is \textit{good} for $\phi$ if it is not crossed by a large contour or more precisely 
$$ \forall \gamma \in \Upsilon^{\larg}(\phi), \quad  \Delta_{\gamma}\cap \cB_z=\emptyset $$
while we say it is \textit{bad} if not (see Figure \ref{goodbad}).

\begin{figure}[ht]
\begin{center}
\leavevmode
\epsfxsize =8 cm
\psfragscanon
\psfrag{N}[c][l]{\small $N$}
\psfrag{M}[c][l]{ \small $M$}
\psfrag{B13}[c][l]{ \tiny $\cB_{(1,3)}$}
\psfrag{bB41}[c][l]{ \tiny $\bar \cB_{(4,1)}$}
\epsfbox{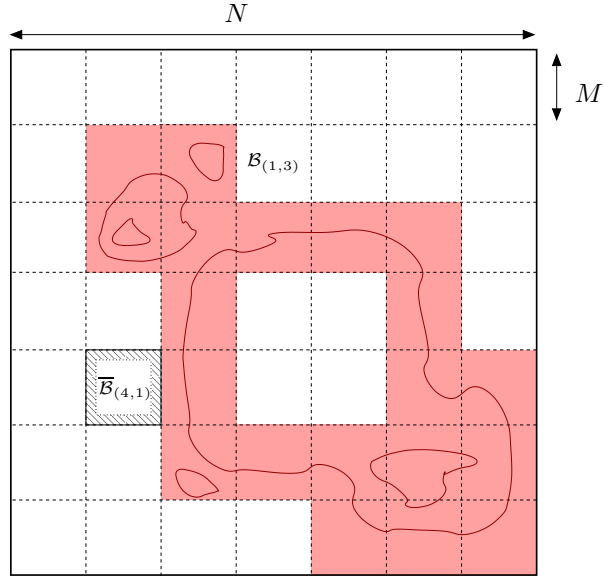}
\end{center}
\caption{\label{goodbad} 
We split our boxes of diameter $N$ into cells of fixed (independent of $N$) diameter $M$.
The set of large contour (of length larger than $|\log u|^2$)
determines the good and bad cells (bad cells are shadowed in the picture).
To avoid boundary effects in our analysis, we consider inside for each cell $\cB_z$ a subcell $\bar \cB_z$ obtained by pushing each side of the cell of an amount
$\sqrt{M}$.
}
\end{figure}

We claim (and this is the crucial point of our proof) 
that conditioned on the set of large contours, we can control the contribution each good cells to the partition function in 
a way that does not depend on the realization of $\phi$
 out of the cell.
For that claim to make sense we need to approximately factorize 
$e^{G^{k,u}(\phi^{-1}[n-k,\infty))}$ to fit the cell decomposition.

\medskip

\noindent From Lemmas \ref{stimaG} and \ref{uppg} we have 
\begin{multline}\label{didicell}
 G^{k,u}(\phi^{-1}[n-k,\infty))\le g^k_1(u)f^{n-k}_1(\phi)+ g^k_2(u)f^{n-k}_2(\phi)-\tilde c_1(\gb)J^{3k}f^{n-k}_{3+}(\phi)\\
 =\sum_{x\in \gL_N} \cX^{n-k}(x)
\end{multline}
where 
\begin{equation}
 \cX^{r}(x):=g^k_1(u) \ind_{\{ q(\phi,x,r)=1\}}+  g^k_2(u)\ind_{\{ 
 q(\phi,x,r)=2\}}- \tilde c_1(\gb) J^{3k}\ind_{\{  q(\phi,x,r)\ge 3\}}.
\end{equation}
An important observation is that for $u$ sufficiently small, $g^k_2(u)$ is negative (recall Lemma \ref{stimaG}).
With our choice of $k$ and using the expression obtained for  $g^k_1(u)$ we thus have 
\begin{equation}\label{1415}
 \cX^{r}(x)\le g^k_1(u)\le 2 u J^{2k_u}\le C u^{14/5}.
\end{equation}
For each $z\in \lint 0,M/N-1\rint^2$ we set  
\begin{equation}
 \bar G^{k,u}(\phi,z)= \sum_{x\in \bar B_z} \cX^{n-k}(x).
\end{equation}
We deduce from \eqref{didicell} that
\begin{equation}\label{groop}
  G^{k,u}(\phi^{-1}[n-k,\infty))\le \sum_{z\in \lint 0, M/N-1 \rint^2 }
  \bar G^{k,u}(\phi,z)+ g^k_1(u)\left|\gL_N \setminus \bigcup_{z\in \lint 0, M/N-1 \rint^2} \bar \cB_z\right|.
\end{equation}
The size of the set in the r.h.s.\ of \eqref{groop} is of order $N^2 M^{-1/2}$ and thus 
 \begin{multline}
 \bar  \tf(\gb,u)=\liminf_{N\to \infty} \frac{1}{N^2}\bE_{N,\gb}\left[ e^{G^{k,u}(\phi^{-1}[n-k,\infty))}\right]
 \\ \le  \liminf_{N\to \infty} \frac{1}{N^2}\bE_{N,\gb}\left[ \prod_{z\in  \lint 0, M/N-1 \rint^2 } e^{\bar G^{k,u}_n(\phi,z)}\right]
 +C M^{-1/2} g^k_1(u).
 \end{multline}
We can check using the definition of $M$ and \eqref{1415} that the last term is of order $u^{3+\gep}$ with $\gep=1/20$. 
Hence to conclude we only need to prove that 
\begin{equation}\label{left}
  \liminf_{N\to \infty} \frac{1}{N^2}\bE_{N,\gb}\left[ \prod_{z\in  \lint 0, M/N-1 \rint^2 } e^{\bar G^{k,u}(\phi,z)}\right]
  \le F(\gb,u)+u^{3+\gep}.
\end{equation}
To do this we need to combine two key results.
The  first allows to estimate the contribution of each cell to the expectation.

\begin{proposition}\label{youbad}
 
 We have for any $z\in \lint 0, M/N-1 \rint^2$, almost surely
 
 \begin{multline}\label{dca}
M^{-2}\log \bE_{N,\gb}\left[ e^{\bar G^{k,u}(\phi,z)}    \ | \ \Upsilon^{\larg}(\phi), \phi\restrict_{\gL_N\setminus \cB_z} \right]\\
\le \begin{cases}
     \left( F(\gb,u)+ u^{3+\gep} \right) \ & \text{ if } \cB_z \text{ is good},\\
     C u^{14/5}    & \text{ if } \cB_z \text{ is bad}.
    \end{cases}
 \end{multline}
 \end{proposition}

This result tells us that good cells give a contribution to the partition function which does not exceed the required bound, but provides only 
a rough bound for bad cells.
To conclude, we need to show that the major contribution to the partition function comes from realization of $\phi$ with few bad cells.
This is the objective our second result which controls the total length of large contours.

\begin{lemma}\label{larg}
 Setting 
 $\cL^{\larg}(\phi):=\sum_{\gamma \in \Upsilon^{\larg}(\phi)} |\tilde \gamma|$,
 if $\gb>\gb_1$ we have
\begin{equation}
\bP_{N,\gb}\left[ \cL^{\larg}(\phi) \ge N^2 u^2  \right]\le e^{-c N^2 u^2}. 
\end{equation}
\end{lemma}

The remainder of the section is organized as follows: in  Section \ref{boodu}, we finish the proof of Proposition \ref{superbien} by proving \eqref{left} using Proposition \ref{youbad} and Lemma \ref{larg}.
These two results are proved respectively in Sections \ref{ybp} and \ref{lastl}.

\subsection{Proof of Proposition \ref{superbien}}\label{boodu}

We prove first that Proposition \ref{youbad} implies that 
\begin{multline}\label{yougood}
 \log \bE_{N,\gb}\left[\prod_{z\in  \lint 0, M/N-1 \rint^2 } e^{\bar G^{k,u}_n(\phi,z)}  \ | \ \Upsilon^{\larg}(\phi)\right]
 \\
 \le M^2  \left(F(\gb,u)+ u^{3+\gep}\right) \cN^N_{\good}+   C M^2 \cN^N_{\bad} u^{14/5},
\end{multline}
where 
$\cN^N_{\good}$ and  $\cN^N_{\bad}$ denote the number of good and bad cells in $\gL_N$.
We notice that if $z$ is fixed, $\bar G^{k,u}(\phi,y)$ is a function of   $\phi\restrict_{\gL_N\setminus \cB_z}$ for all $y\ne z$. 
Thus we have
\begin{multline}
 \bE_{N,\gb}\left[ \prod_{y\in \lint  0 ,M-1 \rint^2} 
 e^{\bar G^{k,u}(\phi,y)}    \ | \ \Upsilon^{\larg}(\phi)\right]
 \\=    \bE_{N,\gb}\left[ \prod_{y\in \lint  0 ,M-1 \rint^2\setminus \{ z\}}  e^{\bar G^{k,u}(\phi,y)} 
   \bE_{N,\gb}\left[ e^{\bar G^{k,u}(\phi,z)} \ | \ \Upsilon^{\larg}(\phi),  \phi\restrict_{\gL_N\setminus \cB_z}\right] \ | \ \Upsilon^{\larg}(\phi) \right]\\
 \le  \bE_{N,\gb}\left[ \prod_{y\in \lint  0 ,M-1 \rint^2\setminus \{z\} }  e^{\bar G^{k,u}(\phi,y)}   \ | \ \Upsilon^{\larg}(\phi) \right]\\
 \times
 \exp\left( \left( F(\gb,u)+ u^{3+\gep}\right)\ind_{\{\cB_z \text{ is good }\}}  + u^{14/5}\ind_{\{\cB_z \text{ is bad }\}} \right).
\end{multline}
Equation \eqref{yougood} then follows by iterating the process.

\medskip

\noindent The bound provided by \eqref{yougood}
is of the right order if  the proportion of bad cells is small.
This observation leads us to make the following decomposition of the partition function:
\begin{multline}\label{deuxterms}
\bE_{N,\gb}\left[\prod_{z\in  \lint 0, M/N-1 \rint^2 } e^{\bar G^{k,u}_n(\phi,z)} \right]= \bE_{N,\gb}\left[\prod_{z\in  \lint 0, M/N-1 \rint^2 } e^{\bar G^{k,u}_n(\phi,z)} \ind_{\{\cL^{\larg}(\phi) \ge N^2 u^2\}}\right]\\
+ \bE_{N,\gb}\left[\prod_{z\in  \lint 0, M/N-1 \rint^2 } e^{\bar G^{k,u}_n(\phi,z)}  \ind_{\{\cL^{\larg}(\phi) < N^2 u^2\}}\right].
\end{multline}
To estimate the first term  \eqref{deuxterms}, we consider the contribution of the 
worse case scenario  where that all boxes are bad.
Using Lemma \ref{larg}, we obtain for $u$ sufficiently small
\begin{multline}\label{lelarge}
  \bE_{N,\gb}\left[e^{G^{k,u}\phi^{-1}[n-k,\infty)}\ind_{\{\cL^{\larg}(\phi) \ge N^2 u^2\}}\right] \\ \le
  e^{CN^2u^{14/5}}\bP_{N,\gb}\left[ \cL^{\larg}(\phi) \ge N^2 u^2  \right] \le  e^{N^2(C u^{14/5}-c u^2)} \le 1.
\end{multline}
The main contribution to the partition function is thus given by the second term in \eqref{deuxterms}.
We know that that the number of bad cells is at most equal to $2\cL^{\larg}(\phi)$. Indeed we the size of the neighborhood 
$|\Delta_{\gamma}|$ is at most equal to $2|\tilde \gamma|$ and each cell contains at least one vertex in the neighborhood of a large contour.
Hence we have 
\begin{equation}
 \{\cL^{\larg}(\phi)<CN^2 u^2\}\subset \{ \cN_{\bad}\le  2CN^2 u^2 \}.
\end{equation}
Using \eqref{yougood} we obtain that 
 \begin{multline}
\bE_{N,\gb}\left[ \prod_{z\in  \lint 0, M/N-1 \rint^2 } e^{\bar G^{k,u}_n(\phi,z)}  \ | \ \cL^{\larg}(\phi)<CN^2 u^2 \right]\\
  \le \exp\left( N^2\left(F(\gb,u)+ u^{3+\gep}\right) + 2C M^2  N^2 u^2 \right)
  \le e^{ N^2\left(F(\gb,u)+ 2u^{3+\gep} \right)}.
 \end{multline}
Together with \eqref{deuxterms} and \eqref{lelarge} this concludes the proof of  \eqref{left}.
 
\qed

\subsection{Proof of Proposition \ref{youbad}}\label{ybp}

The case of bad boxes is trivial as we just need to remark that from \eqref{1415} we have 
$$\bar G^{k,u}(\phi,z)\le M^{2}g^{k}_1(u)\le CM^{2} u^{14/5}.$$
Let us thus focus on good boxes.
With our choice of parameter \eqref{ramets} we have $$|\bar G^{k,u}_n(\phi,z)|\le C M^2 u^{14/5}\le C' u^{9/5}.$$
Hence from Taylor's formula we have
\begin{multline}
\bE_{N,\gb}\left[ e^{\bar G^{k,u}(\phi,z)}  \ | \  \Upsilon^{\larg}(\phi),  \phi\restrict_{\gL_N\setminus \cB_z}\right]\\
\le \bE_{N,\gb}\left[\bar G^{k,u}(\phi,z)  \ | \  \Upsilon^{\larg}(\phi),  \phi\restrict_{\gL_N\setminus \cB_z} \right]
+
 \bE_{N,\gb}\left[ |\bar G^{k,u}(\phi,z)|^2  \ | \  \Upsilon^{\larg}(\phi),  \phi\restrict_{\gL_N\setminus \cB_z} \right]\\
 \le \bE_{N,\gb}\left[\bar G^{k,u}(\phi,z)  \ | \  \Upsilon^{\larg}(\phi),  \phi\restrict_{\gL_N\setminus \cB_z} \right]+ C M^2 u^{23/5}.
\end{multline}
The last term being smaller than $M^2 u^{3+\gep}$, we are done provided we show that
\begin{equation}\label{woops}
 \bE_{N,\gb}\left[\bar G^{k,u}(\phi,z)  \ | \  \Upsilon^{\larg}(\phi), \phi\restrict_{\gL_N\setminus \cB_z} \right]\le M^2 \left(F(\gb,u)+u^{3+\gep}\right).
\end{equation}
This is achieved with the following estimate.
\begin{lemma}\label{dabuty}
If $\cB_z$ is a good box, then for any $x\in \bar \cB_z$ we have 
for any $r\in \bbZ$
\begin{equation} \label{lineq}
 \bE_{N,\gb}\left[  \mathcal X^r(x) \ | \  \Upsilon^{\larg}(\phi), \phi\restrict_{\gL_N\setminus \cB_z} \right]
 \le F(\gb,u)+u^{3+\gep}.
\end{equation}
\end{lemma}
\noindent We obtain \eqref{woops} by applying \eqref{lineq} for $r=n-k$ and summing over $x\in \bar \cB_z$.

\medskip

\noindent We prove the inequality \eqref{lineq} with a slightly stronger conditioning.
We define
\begin{equation}
\Upsilon^{z}(\phi):=\Upsilon^{\larg}(\phi) \cup 
 \left\{ \gamma \in  \Upsilon(\phi) \ : \Xi_\gamma \cap \left( \gL_N \setminus \cB_z \right) \ne \emptyset \right\},
\end{equation}
and let $\hat \Upsilon^{z}(\phi)$ denote the corresponding set of cylinders.
We  are going to prove  
\begin{equation}\label{looze}
\bE_{N,\gb}\left[  \mathcal X^r(x) \ | \  \hat \Upsilon^{z}(\phi) \right]\le  F(\gb,u)+u^{3+\gep}.
\end{equation}
Note that the knowledge of  $\hat \Upsilon^{z}(\phi)$ is sufficient to reconstruct the field outside of   $\gL_N \setminus \cB_z$ (recall \eqref{cylinder}), and thus the inequality 
\eqref{lineq} which involves conditioning with respect to less information  is a consequence of \eqref{looze}.
To show that \eqref{looze} holds (we do so at the end of the section) we need first to provide a description and some key properties for this conditioned measure. 

\medskip

\noindent If $\cB_z$ is good, then the contours of $\Upsilon^{z}(\phi)$ do not touch $\bar \cB_z$. For large contours, this follows from the definition of goodness. The other contours  in 
$\Upsilon^{z}(\phi)$ have diameter smaller than $\ell_u$ we have  
$$d(\bar \cB_z,\cB^{\cc}_z)=\sqrt{M} \ge  \ell_u.$$
Hence we can consider $\gL_z$ the largest connected set remaining in $\cB_z$ after deleting the interior of contours in $\Upsilon^{z}(\phi)$ (see Figure \ref{gamma})
$$\gL_z:=  \bigcup \left\{ \gG \text{ connected sets}  : \  \bar \cB_z\subset \gG \subset \left[\cB_z\setminus \left( \cup_{ \gamma\in \Upsilon^{z}(\phi)}\bar \gamma \right)\right] \right\}.$$
If $u$ is sufficiently small we have 
\begin{equation}
 d(\bar \cB_z,\partial \gL_z)\ge \sqrt{M}- \ell \ge \sqrt{M}/2.
\end{equation}
Using Lemma \ref{geom} and \ref{restrict}, we can describe the law of $\phi\restrict_{\gL_z}$, conditioned to $\hat \Upsilon^{z}(\phi)$:
\begin{itemize}
\item [(A)] The boundary condition is given by
\begin{equation}\label{deffm}
 m=m_z:=\sum_{(\gamma,k)\in \hat \Upsilon^{\larg}(\phi)}\varphi_{(\gamma,k)}(x_z),
\end{equation}
where $x_z$ is an arbitrary point in $\cB_z$ (the function $\sum_{(\gamma,k)\in \hat \Upsilon^{\larg}(\phi)}\varphi_{(\gamma,k)}$ is constant on  $\cB_z$ by definition of good boxes).
\item [(B)]
The distribution of contours in the set $\gL_z$ is the same as that described in Lemma \ref{restrict}, with the additional restrictions that there are no 
large contours and that contours must be compatible with that in $\Upsilon^{z}(\phi)$.
The latter condition can be expressed in a more direct manner. We define
\begin{equation}\begin{split}
 U^+=U^+_z&:= \left( \bigcup_{\gamma \in \Upsilon^{z}_-(\phi)} \Delta^+_{\gamma}  \right)\cap \gL_z,\\
 U^-=U^-_z&:= \left( \bigcup_{\gamma \in \Upsilon^{z}_+(\phi)} \Delta^+_{\gamma}  \right)\cap \gL_z,
\end{split}\end{equation}
where $\Upsilon^{z}_{\pm}(\phi)$ denote the sets of positive/negative contours in $\Upsilon^{z}(\phi)$. 
Note that $U^+$ and $U^-$ are not necessarily disjoint.
We must have 
\begin{equation}
\forall x\in U^-, \  \phi(x)\ge m  \quad \text{and} \quad \forall x\in U^-,\  \phi(x)\le m.
\end{equation}
\end{itemize}

\begin{figure}[ht]
\begin{center}
\leavevmode
\epsfxsize =6 cm
\psfragscanon
\psfrag{ggg}[c][l]{\small $\gL_z$}
\psfrag{barBz}[c][l]{ \small $\bar \cB_z$}
\epsfbox{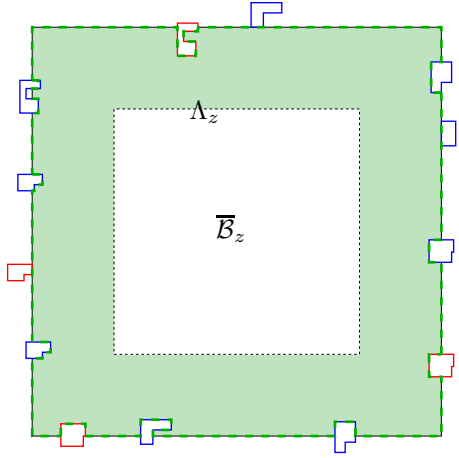}
\end{center}
\caption{\label{gamma} 
In a good cell $\cB_z$ we obtain the domain $\gL_z$ (which is the union of the central white area $\cB_z$ and the shadowed area surrounding it) considering the the connected component containing $\bar\cB_z$ 
after removing the interior of all contours which cross
$\partial \cB_z$. By the assumption that the box is good, these contours have diameter smaller than 
$|\log u|^2$ and thus cannot intersect the inner box $\bar\cB_z$.}
\end{figure}
\noindent We we let $\mu=\mu_z$ denote the  (conditional) distribution of $\phi\restrict_{\gL_z}$ we have thus 
\begin{equation}
 \mu_z(\cdot)= 
 \bP^m_{\gb,\gL_z}\left[\cdot \ | \ \phi\restrict_{U^+}\ge m,\  \phi\restrict_{U^-}\le m, \forall \gamma \in \Upsilon(\phi),\ \Diam(\gamma)\le \ell_u  \right].
\end{equation}
To simplify the notation we set 
\begin{equation}
 \bP_{U^+,U^-}:=\bP^m_{\gb,\gL_z}\left[\cdot \ | \ \phi\restrict_U^+\ge m,  \phi\restrict_U^-\le m \right].
\end{equation}
From now on, we set $m=0$ (this is no loss of generality since changing the value of $m$ corresponds to changing the value of $r$ for which we must prove \eqref{looze}) 
and replace $\gL_z$ by $\gL$, an arbitrary simply connected subset of $\cB_0$ such that $d(\gL^{\cc}, \bar \cB_0)\ge \sqrt{M}/2$.
We also consider $U^+$ and $U^-$ to be arbitrary subsets of $\partial \gL$.

\medskip

We want to show that to estimate the expectation of $\mathcal X(x)$, we can
replace $\mu_z$ by the infinite volume limit $\bP_{\gb}$ at the cost of a very small error term.
We prove this in two steps. Firstly, we show that conditioning on having no large contours only yields a small error term.

\begin{lemma}\label{contourcondi}
When $\gb>\gb_1$
There exists a constant $c$ such that for all $u$ sufficiently small 
 $$\bP_{U^+,U^-}\left(\exists  \gamma \in \Upsilon(\phi),\ \Diam(\gamma)> \ell_u\right)\le e^{-c|\log u|^2}.$$
 As an immediate consequence for any event $A$
 \begin{equation}\label{condition}
 \left|\bP_{U^+,U^-}(A)-  \bP_{U^+,U^-}\left( A \ | \ \forall \gamma \in \Upsilon(\phi),\ \Diam(\gamma)\le \ell_u \right) \right| \le e^{-c|\log u|^2}.
\end{equation}
 
\end{lemma}

Secondly we show that $\bP_{U^+,U^-}$ can be replaced by 
$\bP_{\gb}$ also at the cost of a small error term.
Consider the following small extension of the box $\bar \cB_z$
$$\bar \cB^{(2)}_z:= \{x\in \bbZ^2 \ : \ \exists y\in \bar \cB_z, |y-x|\le 2\}.$$

\begin{lemma}\label{leschaps}
When $\gb>\gb_1$, for any subset  $U^+$  and $U^-$ of $\partial \gL_z$.
For any monotone  event $A\in \sigma\left( \phi\restrict_{\bar \cB^{(2)}_z} \right)$ we have 
\begin{equation}\label{lea}
  |\bP_{U^+,U^-}(A)-\bP_{\gb}(A)|\le e^{-cM^{1/2}}.
 \end{equation}
 The inequality is also valid for the intersection of two monotone events.
 \end{lemma}

\begin{proof}[Proof of Lemma \ref{contourcondi}]

Using the FKG-Harris inequality as in Lemma \ref{restrict}, we find that the distribution of contours for 
 $\bP_{U^+,U^-}$ is dominated by $\bQ_{\gL,\gb}$ (recall \eqref{defgroq}).
 Using this a  union bound we obtain 
 
 \begin{multline}\label{zune}
  \bP_{U^+,U^-}[\exists  \gamma \in \Upsilon(\phi),\ \Diam(\gamma)> \ell_u]\le  \bQ_{\gL,\gb}[\exists  \gamma \in \Upsilon(\phi),\ \Diam(\gamma)> \ell_u]
  \\ \le
  \sum_{\{\gamma \in \cC_{\gL_z} \ : \ \Diam(\gamma)>\ell \}} e^{-\gb|\tilde \gamma|}
  \le M^2 e^{-2c\ell_u}\le e^{-c|\log u|^2/2},
 \end{multline}
where the penultimate inequality is a consequence of $\gb>\gb_1$ (recall \eqref{finiteprim}) provided that $u$ is sufficiently small.

\end{proof}

 \begin{proof}[Proof of Lemma \ref{leschaps}]
The proof for monotone events is the same as that of \cite[Equation (7.20)]{cf:PLMST}. We include it for completeness.
Without loss of generality we can assume that $A$ is increasing.
Using the inequality for monotone events we have 
\begin{equation}
\bP_-(A) \le \bP_{U^+,U^-}(A)\le \bP_+(A)
\end{equation}
where $\bP_\pm$ are obtained by considering the respective cases $(U^+,U^-)=(\partial \gL,\emptyset)$ and  $(U^+,U^-)=(\emptyset,\partial \gL)$.
Hence we only need to prove the result for those measures. 
We consider only $\bP_-$, the proof for $\bP_+$ being obtained by symmetry.

\medskip

Let us consider the set of contours of $\phi$ adjoining the boundary of $\gL$ (note that under $\bP_-$ all these contours have to be negative ones)
$$\Upsilon^{\ext}(\phi):=\left\{\gamma\in \Upsilon(\phi) \ : \ d(\bar \gamma,\partial \gL)=1 \right\}.$$
and set 
$$\gL'(\phi):=\gL \setminus\left( \left(\cup_{\gamma\in \Upsilon^{\ext}(\phi)}\bar \gamma \right)\cup \partial^- \gL \right).$$
where  $\partial^- \gL:=\{ x\in \gL \ : \ \exists y\in \gL^{\cc}, y\sim x \}$ denotes the internal boundary.
Now, we observe that restricted to $\gL'(\phi)$, the distribution of $\phi$ in $\gL'$ dominates $\bP_{\gL',\gb}$.
Indeed the field is conditioned to be non-negative in the neighborhood of $(\bar \gamma)_{\gamma\in \Upsilon^{\ext}}$.

\medskip

\noindent Let us consider the event 
$$E:=\left\{ \left\lint \frac{3\sqrt{M}}{4}, M-\frac{3\sqrt{M}}{4} \right\rint  \subset \gL'\right\}.$$
Because the occurrence of $E^{\cc}$ implies the existence of a contour of size larger than $\sqrt{M}/4$,
we have similarly to \eqref{zune}
\begin{equation}\label{bloor}
\bP^{m}_-(E)\ge 1-e^{-cM^{1/2}}. 
\end{equation}
Now conditioning on $E$ and
 using the stochastic domination mentioned above we have
\begin{equation}
 \bP_-( A \  | \ E)\le \bE^-_{m}\left[ \bP_{\gL',\gb}(A) \ | \ E \right]\le  e^{-c M^{1/2}}+\bP_{\gb}(A),
\end{equation}
where the last inequality uses the exponential decay \eqref{decayz}.
To obtain the corresponding lower bound we also stochastic domination and  \eqref{decayz}. We have
\begin{equation}
  \bP_-( A )\ge \bP_{\gL,\gb}(A)\le  e^{-c M^{1/2}}+\bP_{\gb}(A).
\end{equation}

\medskip

We finally deal with the case of intersection of monotone events.
We consider $C=A\cap B$ where $A$ is increasing and $B$ is decreasing (the other cases are trivial).
We have 
 \begin{multline}
   |\bP_{U^+,U^-}(C)-\bP_{\gb}(C) |\\
  =|\bP_{U^+,U^-}(A)-\bP_{\gb}(A)- \bP_{U^+,U^-}(B^{\cc}\cap A)+\bP_{\gb}(B^{\cc}\cap A)|
  \\
   \le  |\bP_{U^+,U^-}(A)-\bP_{\gb}(A)|+ |\bP_{U^+,U^-}(B^{\cc}\cap A)-\bP_{\gb}(B^{\cc}\cap A)|
   \le 2 e^{-cM^{1/2}},
   \end{multline}
   and we can conclude by changing the value of $c$.
\end{proof}

\begin{proof}[Proof of Lemma \ref{dabuty}]
Recall that we need to prove \eqref{looze}. 
First we remark that using Lemma \ref{contourcondi} and \ref{leschaps}, we can, at the cost of a small error term
replace $\bE_{N,\gb}[ \cdot \ | \ \hat \Upsilon^{z}(\phi)]$ by the infinite volume measure $\bE_{\gb}$
(recalling \eqref{deffm} we need to change $r$ by $r-m$).

\medskip

Indeed as $q(\phi,x,r)$ is an increasing function of $\phi$,
 $\{q(\phi,x,r)\ge 3\}$ is an increasing event. The two other events we have to deal with 
can be written as intersection of two monotone events as $\{q(\phi,x,r)= i\}= \{q(\phi,x,r)\ge i\} \cap \{q(\phi,x,r)\le i\}$.
 To apply Lemma \ref{leschaps} we also use the fact that all these events can be determined by looking at the field in the box $\bar \cB^{(2)}_z$.
Hence to conclude we are left with proving
\begin{equation}\label{dzl}
 \bE_{\gb}\left[ \cX^r(x)  \right] \le F(\gb,u)+u^{3+\gep}.
\end{equation}
We start with the case $r\le \sqrt{k}$.
In that case we have 
\begin{equation}
 \bP_{\gb}\left[  q(x,\phi,r)\ge 2\right]
  \ge  \bP_{\gb}\left[  q(x,\phi,\sqrt{k})\ge 2\right]
 \ge cJ^{3\sqrt{k}}\ge u^{\gep},
 \end{equation}
where the last inequality is valid for any $\gep>0$ provided that $u$ is sufficiently small. 
Replacing $k$ by its value, we obtain that when $r\le \sqrt{k}$
\begin{equation}
 \bE_{\gb}\left[ \cX^r(x)  \right]\le g^1_k(u)- \tilde c_1 J^{3k}\bP_{\gb}\left[  q(x,\phi,\sqrt{k}\ge 2\right]
 \le C u^{14/5}- \tilde c_1 u^{27/10+\gep}\le 0.
 \end{equation}
This case being dealt with, we can assume $r$ to be larger than $\sqrt{k}$. Using Proposition \ref{sendingout} and  Taylor's expansion for  $g^k_1(u)$, $g^k_2(u)$  
to estimate the various terms 
(we use our choice of $k$ to determine which are the dominant terms),
we obtain that 

\begin{multline}
  \bE_{\gb}\left[ \cX^r(x)  \right]\le
  g^k_1(u) \bP_{\gb}\left[     q(x,\phi,r)=1 \right]+
 g^k_2(u) \bP_{\gb}\left[  q(x,\phi,r)=2\right]
 \\
 \le \left(J^{2k} u+C u^2 J^{2k} \right)\left(\alpha_1 J^{2r}+C J^{3r} \right)
    - \left(\frac{J^3-J^4}{1-J^4}J^{3k}- C J^{3k}u  \right)\left( 2 \alpha_2 J^{3r}-C r J^{4r} \right)\\
    \le \alpha_1 J^{2(k+r)} u-  \frac{2\alpha_2(J^3-J^4)}{1-J^4}J^{3(k+r)}
    + C\left( u^2 J^{2(k+r)}+ u J^{2k+3r} + rJ^{3k+4r} \right).
\end{multline}
We notice that the right-hand side is negative if $k+r\le  \frac{|\log u|}{2\gb}-K$ for a fixed $K$ sufficiently large
(it is just sufficient to compare the two first terms as the others are negligible).
For larger values of $r$, the last term is smaller than $u^{3+\gep}$.
We conclude by observing that by definition (recall \eqref{uds})
$$\alpha_1 J^{2(k+r)} u-  \frac{2\alpha_2(J^3-J^4)}{1-J^4}J^{3(k+r)}\le F(\gb,h).$$

\end{proof}

\subsection{Proof of Lemma \ref{larg}}\label{lastl}

For $\gl>0$, $e^{\gl \cL_{u}}$ is an increasing function of the set of contours. Thus we have by Lemma \ref{restrict}
\begin{equation}
\bE_{N,\gb}\left[ e^{\gl \cL_{u}} \right]\le \bQ_{\gL_N,\gb}\left[ e^{\gl \cL_{u}} \right]=
\prodtwo{\gamma\in \cC_{\gL_N}}{|\tilde \gamma|\ge (\log u)^2}
\left(1+e^{-\gb|\tilde \gamma|} \left(e^{\gl|\tilde \gamma|}-1\right) \right).
\end{equation}
Now we remark that as all factors of the products are larger than one. Hence we have
\begin{equation}
 \prodtwo{(\gamma,\gep)\in \cC_{\gL_N}}{|\tilde \gamma|\ge (\log u)^2}\left(1+e^{-\gb|\tilde \gamma|} \left(e^{\gl|\tilde \gamma|}-1\right) \right)
 \le \prod_{x\in \gL_N}\prod_{\{\gamma \ : \ |\tilde \gamma|\ge (\log u)^2 \text{ and } x\in \bar \gamma \}}\!\!\!\!\!\!\!\!\!\!\!\!\!\!\!\!\!\!\!\!\!\left(1+e^{-\gb|\tilde \gamma|} \left(e^{\gl|\tilde \gamma|}-1\right) \right),
\end{equation}
As a consequence we obtain
\begin{equation}
 \bE_{N,\gb}\left[ e^{\gl \cL_{u}} \right]\le 
 \left( \prod_{\{\gamma \ : \ |\tilde \gamma|\ge (\log u)^2 \text{ and } {\bf 0}\in \bar \gamma \}}1+e^{(\gl-\gb)|\tilde \gamma|}\right)^{N^2},
\end{equation}
which yields in turn
\begin{equation}
\frac{1}{N^2} \log  \bE_{N,\gb}\left[ e^{\gl \cL_{u}} \right]\le 
 \sum_{\{\gamma \ : \  |\tilde \gamma|\ge (\log u)^2 \text{ and } {\bf 0}\in \bar \gamma \}}e^{(\gl-\gb)|\tilde \gamma|}.
\end{equation}
If $\gb>\gb_1$ (recall \eqref{finiteprim}) then choosing $\gl_{\gb}= \frac{\gb-\gb_1}{2}$ the right-hand side is smaller than
$e^{-\frac{(\gb-\gb_1) (\log u)^2}{4}}$ for small values of $u$.
Using the usual large deviation computation, we can conclude that
\begin{equation}
 \frac{1}{N^2} \log  \bE_{N,\gb}\left[ \cL_{u} \ge  u^2  \right]\le -\gl_{\gb} u^2 +e^{-\frac{(\gb-\gb_1) (\log u)^2}{4}}
 \le -\gl_{\gb} u^2/2.
\end{equation}

\bigskip

{\bf Acknowledgements:} The author is grateful to Fabio Toninelli for enlightening suggestions. He also acknowledges the support of a productivity grant from CNPq and
of a JCNE  grant from FAPERJ.

\appendix

\section{Estimating peak probabilities}\label{ppeaks}

\subsection{Proof of Proposition \ref{rourou}}

Let us start with the lower bound part.
Given $A\subset \gL$ we define  $\Theta_A$ the bijection from $\gO_{\gL}$ to itself which adds one to every coordinate in $A$, that is
$$\Theta_A(\phi)=\phi+\ind_A$$
It is easy to check that given $x,y$ and $z$ satisfying the assumption $x\sim y \sim z$ we have 
\begin{multline}\label{lezgop}
 \cH_{\gL}(\Theta_{\{x\}}(\phi))\le  \cH_{\gL}(\phi)+4, \quad  \cH_{\gL}(\Theta_{\{x,y\}}(\phi))\le  \cH_{\gL}(\phi)+6  \\  
 \text{ and }  \ \  \cH_{\gL}(\Theta_{\{x,y,z\}}(\phi))\le  \cH_{\gL}(\phi)+8.
 \end{multline}
As $\Theta^n_{\{{x}\}}$ maps $\{\phi(x)\ge 0 \}$ to   $\{\phi(x)\ge n \}$, using the above we have 
\begin{multline}\label{lezgoop}
 \bP_{\gL,\gb}[\phi(x)\ge n]= \frac{1}{\cZ_{\gL,\gb}}\sum_{\{ \phi\in \gO_{\gL} \ : \ \phi(x)\ge 0\}} e^{-\gb \cH_{\gL}(\Theta^n_{\{x\}}(\phi))}
 \\ \ge e^{-4n\gb}  \frac{1}{\cZ_{\gL,\gb}}\sum_{\{ \phi\in \gO_{\gL} \ : \ \phi(x)\ge 0\}} e^{-\gb \cH_{\gL}(\phi)}
 = e^{-4n\gb} \bP_{\gL,\gb}[\phi(x)\ge 0].
\end{multline}
Then as $\bP_{\gL,\gb}[\phi(x)\ge 0]+\bP_{\gL,\gb}[\phi(x)\le 0]\ge 1$ we conclude by symmetry that 
\begin{equation}\label{dehalf}
\bP_{\gL,\gb}[\phi(x)\ge 0]\ge 1/2. 
\end{equation}
The same computation gives 
\begin{equation}
\begin{split}
  \bP_{\gL,\gb}[\min(\phi(x),\phi(y))\ge n]&\ge e^{-6n\gb}\bP_{\gL,\gb}[\min(\phi(x),\phi(y))\ge 0],\\
    \bP_{\gL,\gb}[\min(\phi(x),\phi(y),\phi(z))\ge n]&\ge e^{-8n\gb}\bP_{\gL,\gb}[\min(\phi(x),\phi(y),\phi(z))\ge 0].
    \end{split}
\end{equation}
We conclude observing that by \eqref{dehalf} and the FKG inequality \eqref{FKG} we have  
\begin{equation}
\bP_{\gL,\gb}[\min(\phi(x),\phi(y))\ge 0 ] \ge \frac{1}{4} \quad  \text{ and } \quad  \bP_{\gL,\gb}[\min(\phi(x),\phi(y),\phi(z))\ge 0]\ge \frac{1}{8}.
\end{equation}
Concerning the upper bound, we choose to treat in detail only the case of $\{x,y,z\}$ which is the most delicate of the three.
If all $\phi(x)$, $\phi(y)$ and $\phi(z)$ are larger than $n$, then there exists a sequence of positive contours $\gamma_1,\dots,\gamma_k$ satisfying (with strict inclusion)
$$\{x,y,z\} \subset \bar\gamma_1 \subset \dots \subset \bar\gamma_k $$
and $m_1,\dots,m_k$ satisfying  $\sum_{i=1}^k m_k= n$ such that for all $i\in \lint 1, k \rint$, 
$\gamma_i\in \Upsilon(\phi)$ and has intensity (recall \eqref{lintens} larger than $m_i$.
From Lemma \ref{geom} and \ref{restrict} we have
\begin{equation}
 \bP_{\gL,\gb}[\forall i \in \lint 1, k\rint, \  \gamma_i\in \Upsilon(\phi) \text{ and has intensity larger than $m_i$ }]\le \prod_{i=1}^k e^{-m_i\gb|\tilde\gamma_i|}.
 \end{equation}
Using union bound, we have thus 
\begin{multline}\label{lebron}
 \bP_{\gL,\gb}[\min(\phi(x),\phi(y),\phi(z))\ge n]\\
 \le \sum_{k=1}^n \sum_{\{ (m_i)_{i=1}^k\ : \ \sum_{i=1}^k m_i= n \}} 
 \sum_{\{x,y,z\} \subset \bar\gamma_1\subset\dots\subset\bar \gamma_k} e^{- \sum_{i=1}^k m_i\gb|\tilde \gamma_i|},
 \end{multline}
 where the last sum ranges over $k$-tuples of geometric contours $\tilde \gamma_1,\dots,\tilde\gamma_k$ satisfying the strict inclusion condition.
By elementary geometric consideration, we have $|\tilde \gamma_1|,|\tilde \gamma_2|\ge 8$ and as the other contours enclose at least $5$ vertices 
$|\tilde \gamma_i|\ge 10$ for $i\ge 3$.
 Hence
 \begin{multline}\label{james}
 \sum_{i=1}^k m_i|\tilde \gamma_i|= \sum_{i=1}^k|\tilde \gamma_i|+\sum_{i=1}^2 8(m_i-1)+\sum_{i=3}^k 10(m_i-1)\\
 \ge \sum_{i=1}^k |\tilde \gamma_i|+ 10(n-k)-2(m_1+m_2).
 \end{multline}
 Using \eqref{finite} there exists $K$ such that 
\begin{equation}
 \sum_{\{\tilde \gamma \ : \ |\bar \gamma|\ge K \}} e^{-\gb|\tilde \gamma|}\le e^{-10\gb |\tilde \gamma|}.
\end{equation}
Replacing the condition $\{x,y,z\} \subset \bar \gamma_1\subset\dots\subset\bar \gamma_k$ by the less restrictive one $|\bar\gamma_i|\ge K$ if $i\ge K-2$, we obtain combining \eqref{lebron}
and \eqref{james}
(for a constant $C$ depending on $K$ and $\gb$)
\begin{multline}
  \bP_{\gL,\gb}[\min(\phi(x),\phi(y),\phi(z))\ge n]\\
  \le  \sum_{k=1}^n 
  \left( \sum_{\{\tilde\gamma \ : \ \{x,y,z\} \subset  \bar \gamma \}}\ e^{-\gb \tilde\gamma|}   \right)^{\max(k,K-2)}
  \left( \sum_{\{\tilde\gamma \ : \ |\bar \gamma|\ge K \}} e^{-\gb|\tilde \gamma|} \right)^{(k-(K-2))_+} \\
  \sum_{\{ (m_i)_{i=1}^k\ : \sum_{i=1}^k m_i= n \}} e^{-\gb\left( 10(n-k)-2(m_1+m_2) \right)}
 \\ \le C \sum_{k= 1}^n  \sum_{\{ (m_i)_{i=1}^k\ : \sum_{i=1}^k m_i= n \}}  e^{-\gb\left( 10n-2 (m_1+m_2) \right)},
\end{multline}
with the convention that $m_2=0$ when $k=1$.
Now to estimate the last sum, we notice that  for $k\ge 2$, given $m_1+m_2=q$ there are $2^{n-m-q}$ ways to choose $k$ and $(m_3,\dots,m_k)$, and 
$q-1$ ways to choose $(m_1,m_2)$.
Considering also the case $k=1$, as \eqref{finite} implies that $e^\gb>2$, we obtain that 
\begin{multline}
   \bP_{\gL,\gb}[\min(\phi(x),\phi(y),\phi(z))\ge n]
   \\\le Ce^{-8\gb n}\left( 1 + \sum_{q=2}^n (m-1) 2^{n-q} e^{-2\gb (n-q)} \right)
   \le C e^{-8\gb m} m.
 \end{multline}
The proof for $\{x,y\}$ and $\{x\}$ follow the same scheme, except that in those case we have only one contour of minimal length:
$|\tilde \gamma_1|\ge 6$ and $|\tilde \gamma_2|\ge 8$ in the first case and $|\tilde \gamma_1|\ge 4$ and $|\tilde \gamma_2|\ge 6$ in the second one.

\qed

\subsection{Proof of Proposition \ref{sendingout}}

We can restrict to proving only the second line in \eqref{qtronz} since the proof of the first line follows the same ideas and is simpler.
Set 
$$\cA_n:= \{ \min\left(\phi(x),\phi(y)\right)\ge n \}$$
Note that we have 

\begin{equation}\label{jove}
 \left|\bP_{\gL,\gb}\left(\cA_n\right)-\alpha_2J^{3n} \right|
 \le |\bP_{\gL,\gb}(\cA_n)- \bP_{\gb}(\cA_n) |+ 
 \left|\bP_{\gb}(\cA_n)-\alpha_2J^{3n} \right|.
 \end{equation}
From \eqref{decayz}, the first term  decays exponentially with $d(x,\partial \gL)$.
Hence we only need to control the second term.
Using \eqref{lezgop}, we obtain similarly to \eqref{lezgoop}
\begin{equation}
 \bP_{\gL,\gb}\left[\cA_{n+1}\right]= \bP_{\gL,\gb}\left[\Theta_{\{x,y\}}(\cA_{n})\right]\ge e^{-6\gb} \bP_{\gL,\gb}\left[\cA_n\right].
\end{equation}
Passing to the limit when $\gL\to \infty$ we obtain that
$e^{6\gb n}\bP_{\gb}\left[ \cA_n \right]$ is a non-decreasing  sequence. According to Proposition \eqref{rourou} (recall that the bounds are uniform in $\gL$)
it is also bounded and thus the limit $\alpha_2$ is well defined and for every $n$ we have
\begin{equation}\label{thescore}
e^{6\gb n}\bP_{\gb}\left[ \cA_n \right]\le \alpha_2.
\end{equation}
Let us now prove the lower bound.
Note that \eqref{lezgop} is an equality on the event $\cA_n\cap \cB_n$ where
$$\cB_{n}:= \left\{\forall z \in \partial \{x,y\}, \ \phi(z)\le n \right\}.$$
Using this information we have 
\begin{equation}\label{good}
   \bP_{\gb}\left[ \cA_{n}\right]  \ge \bP_{\gb}\left[ \cA_{n}\cap \cB_{n}\right]=
   e^{6\gb}\bP_{\gb}\left[ \cA_{n+1}\cap \cB_{n}\right]\\
  \ge   e^{6\gb}\bP_{\gb}\left[ \cA_{n+1} \right]-Cn J^{4n}.
\end{equation}
where to obtain the last inequality we used Proposition \ref{rourou} and a union bound in the following manner
\begin{equation}
\bP_{\gb}\left[ \cA_{n+1}\cap \cB_{n}^{\cc}\right]
\le \bP_{\gb}\left[ \bigcup_{z\in \partial \{x,y\}} \min(\phi(x),\phi(y),\phi(z))\ge n \right]\le Cn J^{4n}.
\end{equation}
Iterating \eqref{good} we obtain 
\begin{equation}
  \bP_{\gb}\left[ \cA_{n}\right]  \ge  \lim_{k\to \infty} e^{6\gb k}\left[ \cA_{n+k} \right]
  -C \sum_{m=n}^{\infty} m J^{4m}
  \ge  e^{-6\gb n} \alpha_2-C' n  J^{4n}.
\end{equation}
We now move to the proof of  \eqref{qtronzac}.
Similarly to \eqref{jove} we can restrict to the infinite volume case.
We perform the full proof only in the more delicate case $\{q(\phi,x,n)=2\}$ the other being only easier.
Setting $$\cA_n(A):=\{ \forall x\in A, \phi(x)\ge n\},$$
we have
\begin{equation}\label{subs}
   \{q(\phi,x,n)=2\} \supset 
   \left( \bigcup_{y\sim x } \cA_n \{x,y\} \right) \setminus \left( \bigcup_{\{(y,z) \ : \ \{x,y,z\} \text{is connected} \}  } \cA_n \{x,y,z\} \right).
   \end{equation}
A simple union bound and \eqref{thescore} gives 
\begin{equation}
  \bP_{\gb}(q(\phi,x,n)=2) \le \bP_{\gb}\left( \bigcup_{y\sim x } \cA_n \{x,y\}\right)\le 4\alpha_2 e^{-6\gb n}.
\end{equation}
On the other hand, using inclusion exclusion  and \eqref{qtronz}, we obtain 
\begin{multline}
\bP_{\gb}\left( \bigcup_{y\sim x } \cA_n \{x,y\}\right)\ge \sum_{y\sim x} \bP_{\gb}\left(\cA_n \{x,y\}\right)-  
\sumtwo{y_1,y_2\sim x}{y_1\ne y_2} \bP_{\gb}\left(\cA_n \{x,y_1,y_2\}\right)\\
\ge 
4\alpha_2 e^{-6\gb n}- C n e^{-8\gb n}.
\end{multline}
Finally also by \eqref{qtronz} the probability of the event subtracted event on the l.h.s of \eqref{subs} is also small
\begin{equation}
 \bP_{\gb}\left(\bigcup_{\{(y,z) \ : \ \{x,y,z\} \text{is connected} \}  } \cA_n \{x,y,z\}\right)\le C n e^{-8\gb n},
\end{equation}
which allows to conclude.

\qed

\end{document}